%% file: main.tex
\documentclass[manuscript]{acmart}
\usepackage{multirow}
\usepackage{tikz}
\theoremstyle{definition}

\newtheorem{prop}{Proposition}[section]

\newtheorem{fact}[prop]{Fact}
\newtheorem{remark}[prop]{Remark}

\input{macros}

\DeclareMathOperator{\Pred}{Pred}

\AtBeginDocument{%
  }

\setcopyright{cc}
\setcctype{by}
\acmJournal{TEAC}
\acmYear{2026} \acmVolume{1} \acmNumber{1} \acmArticle{}
\acmMonth{1} \acmDOI{10.1145/3839240}
\acmISBN{978-1-4503-XXXX-X/2018/06}

\title[Bridging the Gap Between Stable Marriage and Stable Roommates]{Bridging the Gap Between Stable Marriage and Stable Roommates: A Parameterized Algorithm for Optimal Stable Matchings}

\author{Christine T. Cheng}
\email{ccheng@uwm.edu}
\orcid{0000-0003-0524-5368}
\affiliation{%
  \institution{University of Wisconsin, Milwaukee}
  \city{Milwaukee}
  \state{WI}
  \country{USA}
}

\author{Will Rosenbaum}
\email{W.Rosenbaum@liverpool.ac.uk}
\orcid{0000-0002-7723-9090}
\affiliation{%
  \institution{University of Liverpool}
  \city{Liverpool}
  \country{UK}
}



\received{20 February 2007}
\received[revised]{12 March 2009}
\received[accepted]{5 June 2009}

\begin{document}

\begin{abstract}
	In the Stable Roommates Problem (SR), a set of $2n$ agents rank one another in a linear order. The goal is to find a matching that is stable: one that has no pair of agents who mutually prefer each other over their assigned partners. We consider the problem of finding an {\it optimal} stable matching. Agents associate weights with each of their potential partners, and the goal is to find a stable matching that minimizes the sum of the associated weights. Efficient algorithms exist for finding an optimal stable matching in the Stable Marriage Problem (SM), but the problem is NP-hard for general SR instances.

	In this paper, we define a notion of structural distance between SR instances and SM instances, which we call the \emph{minimum crossing distance}. When an SR instance has  minimum crossing distance $0$, the instance is structurally equivalent to an SM instance, and this structure can be exploited to find an optimal stable matching efficiently. More generally, we show that when an SR instance has minimum crossing distance $k$, an optimal stable matching can be computed in time $2^{O(k)} n^{O(1)}$. Thus, the optimal stable matching problem is fixed-parameter tractable (FPT) with respect to minimum crossing distance.
\end{abstract}

\maketitle

\input{intro}

\input{SM}

\input{SR}
\input{MedianSemilattices}
\input{MaximalElements}
\input{min-crossings}
\input{conclusions}

\urlstyle{same}
\bibliographystyle{plainnat}
\bibliography{sources}

\appendix

\input{Appendix}

\input{mvc-hardness}

\end{document}

%% file: macros.tex
\newcommand{\dft}[1]{\textit{#1}}

\newcommand{\abs}[1]{\left|#1\right|}

\newcommand{\set}[1]{\left\{#1\right\}}

\newcommand{\sucht}{\,\middle|\,}

\newcommand{\calP}{\mathcal{P}}

\DeclareMathOperator{\mco}{MCO}

\renewcommand{\descriptionlabel}[1]{{\hspace\labelsep\normalfont\itshape#1}}

\let\mc\mathcal
\let\ol\overline
\let\sq\sqsubseteq


%% file: intro.tex
\section{Introduction}\label{sec:introduction}

In their seminal work, Gale and Shapley~\cite{Gale1962College} introduced the Stable Marriage Problem (SM), as well as its non-bipartite version,
the Stable Roommates Problem (SR). In SM, a set of agents is partitioned into two groups, traditionally referred to as men and women, where each agent ranks members of the opposite group in a strict linear order. The goal is to find a matching that is {\it stable}, one that has
no pair of agents who
mutually prefer each  other over their assigned partners.
The main result in~\cite{Gale1962College} provides an efficient algorithm for finding such a matching, thus establishing that {\it every} SM instance has a stable matching.

In SR, there is a single set of agents, each of whom ranks all other agents in a strict linear order. The notion of stability extends naturally to this setting. Unlike SM, however, Gale and Shapley showed that SR instances need not have a stable matching.    Irving~\cite{Irving1985efficient} later devised an efficient algorithm for SR solvability, which was recently shown to be optimal~\cite{Rosenbaum2025quadratic}.



Conway~\cite{Knuth1976MariagesStables, Knuth1997Stable} was the first to observe that when $I$ is an SM instance, its set of stable matchings forms a {\it distributive lattice} $\mc{D}(I)$. The minimum element of $\mc{D}(I)$  is the  {\it man-optimal stable matching} while its maximum element is the {\it women-optimal} stable matching. All other stable matchings of $I$ lie between these two extreme matchings,  where stable matchings that are better for one group are  worse for the other.

Related to $\mc{D}(I)$ is the {\it rotation poset} $\mc{R}(I)$.  Its elements are called {\it rotations} and represent the minimal differences between the stable matchings of $I$.  Irving and Leather \cite{Irving1986Complexity} proved that the stable matchings of $I$ are in one-to-one correspondence with the {\it closed subsets} of $\mc{R}(I)$.  Furthermore,
$\mc{R}(I)$ has the advantage over  $\mc{D}(I)$ in that it can be derived directly from the preference lists of $I$ in linear time.  For this reason, rotation posets play a central role in SM.

Rotation posets have been employed to obtain hardness results for many SM problems such as counting stable matchings~\cite{Irving1986Complexity}, and computing median stable matchings~\cite{Cheng2008Generalized,Cheng2010Understanding}, balanced stable matchings~\cite{Feder1995Stable}, and sex-equal stable matchings~\cite{Kato1993Complexity,McDermid2014SexEqual}.
On the other hand, they have also been used to devise efficient algorithms that compute fair or nice stable matchings like egalitarian stable matchings~\cite{Irving1987efficient,Feder1989new,Gusfield1989Stable} and the more general optimal stable matchings~\cite{Irving1987efficient, Gusfield1989Stable,Gonczarowski2019Stable}, center stable matchings~\cite{ChengMS11, ChengMS16} and robust stable matchings~\cite{ChenSS19}.

SR can be viewed as a generalization of SM because an SM instance can be transformed into an 
SR instance by appending all remaining agents to each agent's preference lists arbitrarily. It is straightforward to show that this transformation will result in an SR instance with precisely the same stable matchings as the original SM instance.

Like SM, a solvable SR instance $I$ also has a rotation poset $\mc{R}(I)$.   Gusfield~\cite{Gusfield1988Structure} and Irving~\cite{Irving1986Stable} (cf.~\cite{Gusfield1989Stable}) showed that $\mc{R}(I)$ consists of singular and non-singular rotations, and the subposet formed by the non-singular rotations $\mc{R}'(I)$ is the significant one.  There is a one-to-one correspondence between the stable matchings of $I$ and the {\it complete closed subsets} of $\mc{R}'(I)$.  Moreover, $\mc{R}(I)$ (and therefore  $\mc{R}'(I)$) can be derived from the preference lists of $I$ in polynomial time.  Much later, Cheng and Lin \cite{Cheng2011Stable} generalized the structure of $\mc{R}'(I)$ to {\it mirror posets} to show that the stable matchings of $I$ form a {\it median graph} $G(I)$, a well-studied class of graphs that includes trees, hypercubes, grids and the cover graphs of distributive lattices (e.g., see \cite{Klavzar1999Median, Imrich2000Product, Bandelt2008Metric}).  For the rest of the section, we shall assume $I$ is a solvable SR instance and refer to $\mc{R}'(I)$ as the {\it mirror poset of $I$}.  We summarize the correspondence between analogous concepts in SM and SR in Figure~\ref{fig:sm-sr-table}.

It is notable that  median graphs---the meta-structures of SR stable matchings---are not ordered while distributive lattices--- the meta-structures of SM stable matchings---are ordered.  In SM, the agents are divided into two groups and this bipartition defines the two ``sides" of the distributive lattice.  In SR, no such partition exists for the agents so no natural ordering can be used on the stable matchings.  Nonetheless, it is possible to transform the median graph $G(I)$ into an ordered structure by designating one of the stable matchings of $I$, $\eta$, as the {\it root} and directing all edges away from $\eta$.
The result is the Hasse diagram of the {\it median semilattice} $\mc{L}(I, \eta)$ \cite{Cheng2011Stable}, a more general structure than a distributive lattice.  The semilattice $\mc{L}(I, \eta)$ has $\eta$ as its minimum element, but it can have one or more maximal elements.
Moreover, different roots will  result in different median semilattices, some of which may prove more algorithmically advantageous than others.



In spite of the similarities between the structures formed by the SM and SR stable matchings,
we do not know of any work that utilizes the mirror posets of SR instances to address computational problems in SR.\footnote{In Section 4.4.3 of \cite{Gusfield1989Stable},  Gusfield and Irving describe a procedure for finding a  {\it minimum regret stable matching} of an SR instance $I$. It involves picking the ``right" rotations to eliminate in the second phase of Irving's algorithm.  While the procedure makes use of rotations, it does so without being aware of the poset formed by the rotations.}
Part of the reason is that SR instances contain SM instances. Thus, when a problem is computationally hard in SM, the analogous problem is also  hard in SR. Yet there are problems such as finding an egalitarian, a rank-maximal,  or a generous stable matching  and their generalization, the optimal stable matching, that can be solved efficiently in SM but are NP-hard in SR~\cite{Irving1987efficient, Feder1989new, Feder1992journal, CooperM20}.
Intuitively,  the more structurally similar an SR instance is to
an SM instance, the easier it should be to find an optimal stable matching.  This intuition is our motivation for developing a fixed-parameter tractable (FPT) algorithm for optimal SR stable matchings.  The parameter we employ, {\it minimum crossing distance}, captures how close---i.e., structurally similar---a given SR instance is to an SM instance.

\begin{figure}
	\begin{center}
		\begin{tabular}{c|c|c}
			Concept                                        & SM Instances                & SR Instances                              \\[2ex]
			\hline
			structure of stable matchings                  & distributive lattice        & median graph                              \\[1ex]
			structure of (non-singular) rotations          & arbitrary poset $\calP$     & mirror poset $\calP$                      \\[1ex]
			\multirow{2}{*}{ } relationship between stable & stable marriages  $\simeq$  & stable matchings $\simeq$ \emph{complete} \\
			matchings and rotation sets                    & closed subsets of rotations & closed subsets of rotations               \\[1ex]
			\multirow{2}{*}{ } complexity of optimal       & polynomial time             & NP-Hard                                   \\
			stable matching
		\end{tabular}
	\end{center}
	\caption{Relationship between structure and terminology for SM and SR instances.\label{fig:sm-sr-table}}
	\Description{A three-column table compares stable marriage and stable roommates instances. It identifies their stable-matching structures as distributive lattices and median graphs, respectively; their rotation structures as arbitrary and mirror posets; their rotation-set correspondence; and the polynomial-time versus NP-hard complexity of optimal stable matching.}
\end{figure}


\paragraph{Our Contributions}

In the {\it optimal stable matching} problem, we are given a cost function $C$ so that $C(a,b)$ is the cost to agent $a$ of being matched to agent $b$.  The cost of stable matching $\mu$ is $C(\mu) = \sum_{\{a,b\}\in \mu} C(a,b) + C(b,a)$.  The goal is to find a stable matching with least cost.  When $C(a,b)$ is equal to the rank of $b$ in $a$'s preference list, an optimal stable matching is also called an {\it egalitarian stable matching}.


The {\it profile} of $\mu$ is a vector $\langle p_1, p_2, \hdots, p_n \rangle$,  where $p_i$ is the number of agents matched to their $i$th choice.  The {\it reverse profile} of $\mu$ is $\langle p_n, p_{n-1}, \hdots, p_1 \rangle$.  A stable matching is {\it rank-maximal} if its profile is lexicographically maximum over all stable matchings---that is, it has the most agents matched to their first choice, then subject to that, their second choice and so on.
A stable matching is {\it generous} if its reverse profile is lexicographically minimum over all stable matchings---that is, it has the least number of agents matched to their last choice, then subject to that, the second-to-last choice and so on.  By choosing appropriate weights, both rank-maximal and generous stable matchings can be viewed as optimal stable matchings \cite{Irving1987efficient, CooperM20}.


In Theorem~\ref{thmLocalOpt}, we show that the technique used by Irving et al.~\cite{Irving1987efficient} to  compute an optimal SM stable matching efficiently  can also be used to find a {\it locally optimal} stable matching of an SR instance $I$.  That is, suppose the median graph $G(I)$ is rooted at the stable matching $\eta$ to create the median semilattice $\mc{L}(I, \eta)$. For any stable matching $\mu$, denote as $[\eta, \mu]$ the set that contains $\mu$ and its predecessors in $\mc{L}(I, \eta)$.  It turns out that the subposet induced by  $[\eta, \mu]$ forms a distributive lattice.  Using Irving et al.'s technique,  a stable matching with the least cost in $[\eta, \mu]$ can be computed in $O(n^4 \log n)$ time when $I$ has $2n$ agents.


\paragraph{Our algorithm} Let the {\it maximal} elements of $\mc{L}(I, \eta)$ be $\mu_1, \mu_2, \hdots,  \mu_r$.  Notice that $\cup_{i=1}^r$ $[\eta, \mu_i]$ contains all the stable matchings of $I$. Moreover, every optimal stable matching of $I$ belongs to some $[\eta, \mu_i]$, where it is also a local optimum.  Thus,
we can obtain an optimal stable matching of $I$ as follows:  first, compute a locally optimal stable matching of $[\eta, \mu_i]$ for $i = 1, 2, \hdots, r$. Then, among the $r$ locally optimal stable matchings,  return the one with the least cost.  By Theorem~\ref{thmLocalOpt}, the running time of the algorithm is $O(r \times n^4 \log n)$.  Clearly, the fewer the maximal elements of $\mc{L}(I, \eta)$, the faster is the algorithm.  In particular, when $\mc{L}(I, \eta)$ has only one maximal element---and therefore the maximal element is a {\it maximum} element---$\mc{L}(I, \eta)$ is actually a distributive lattice and the  algorithm we have just described reduces to Irving et al.'s algorithm.


How then should we choose a root for $G(I)$ so that the resulting median semilattice has the fewest number of maximal elements?
To answer this question, we make use of the mirror poset $\mc{R}'(I)$, a structure that can be built efficiently from the preference lists of $I$.  We show that rooting $G(I)$ at the stable matching $\eta$ has an analogous  action on $\mc{R}'(I)$:  {\it orienting $\mc{R}'(I)$ at $S_\eta$}, where $S_\eta$ is the complete closed subset of $\mc{R}'(I)$ that corresponds to $\eta$.  The orientation partitions $\mc{R}'(I)$ into two parts.  The lower part $\mc{R}^-$ consists of the rotations in $S_\eta$ while the upper part $\mc{R}^+$ contains the rest of the rotations.  We call an edge of the Hasse diagram of $\mc{R}'(I)$ a {\it crossing edge} if it  has one endpoint  in $\mc{R}^-$ and another endpoint in $\mc{R}^+$.   In Theorem \ref{thmmaximalrep},
we show that the maximal elements of $\mc{L}(I, \eta)$ can be described using the crossing edges of $\mc{R}'(I)$.  Consequently, we prove in Corollary \ref{cor:crossing-maximal} that when $\mc{R}'(I)$ oriented at $S_\eta$ has $2k$ crossing edges, the corresponding median semilattice $\mc{L}(I, \eta)$ has at most $3^k$ maximal elements.

In light of Corollary \ref{cor:crossing-maximal},  we investigate the following optimization problem:

\noindent \textsc{Minimum Crossing Orientation (MCO):} {\it Given a mirror poset $\mc{P}$, find an orientation of $\mc{P}$ that has the least number of crossing edges among all the orientations of $\mc{P}$.}


Let MCO($\calP$)  denote the number of crossing edges in an optimal orientation of $\calP$.  When $\calP = \mc{R}'(I)$, we also refer to MCO($\calP$) as the {\it minimum crossing distance of $I$} because it measures how far the SR instance $I$ is from an SM instance.  In particular, we prove in Lemma \ref{prop:MCO} that  MCO($\calP$) $= 0$ if and only if $G(I)$ is the cover graph of a distributive lattice.
Corollary \ref{cor:crossing-maximal} generalizes this result further by implying that if MCO($\calP$) $= k$ then $G(I)$ is the union of the cover graphs of at most $3^k$ distributive lattices.

In Theorem \ref{thm:mco-hard}, we prove that MCO is NP-hard via a reduction from the minimum vertex cover problem.  Next, by doing a reduction to Almost 2-SAT, we show in Theorem \ref{thm:mco} that if MCO($\calP$) $=k$ then MCO can be solved in time $O(15^k k^4 m^3)$, where $m$ is the number of edges in the Hasse diagram of $\calP$.  Applying this FPT result to $\mc{R}'(I)$ and combining it with our algorithm  for finding an optimal stable matching of $I$,  we have our main result.

\begin{theorem}\label{thm:main}
	Let $I$ be an SR instance with $2n$ agents and minimum crossing distance $k$. An optimal stable matching for $I$ can be found in time $2^{O(k)} n^{O(1)}$. Thus, the optimal stable matching problem is fixed-parameter tractable with respect to the minimum crossing distance of $I$.
\end{theorem}


Theorem \ref{thm:main} serves as a bridge between the polynomial-solvability of optimal stable matchings in the SM setting and the NP-hardness of the same problem in the SR setting by showing that the closer an SR instance is to an SM instance, the faster it is find an optimal stable matching.

We believe that the approach we have taken in this paper can be applied to other problems in stable matchings as well, transforming (polynomial-time or FPT) algorithms in SM to FPT algorithms in SR so that the more similar an SR instance is to an SM instance, the faster is the algorithm.
Beyond stable matchings,  we also suspect that the results here are generalizable to other combinatorial objects that form distributive lattices in a classical setting (see \cite{Propp1997Generating} and \cite{Felsner2004Lattice} for examples)
but form median graphs in a more general one.

In Sections 2 and 3, we present the stable marriage (SM) and stable roommates (SR) problems and their associated structures, respectively,  focusing on their similarities and differences.  We outline Irving et al.'s technique for computing an optimal SM stable matching in Section 2 and show how it can be generalized to the SR setting in Section 3.  In Section 4,  we discuss median semilattices and how they can be analyzed using oriented mirror posets.    In Section 5, we connect the crossing edges of an oriented mirror poset to the maximal elements of the corresponding median semilattices.  Finally,  we introduce the minimum crossing orientation problem in Section 6, prove its hardness, present its FPT algorithm and conclude with a proof of Theorem \ref{thm:main}, our main result.

\paragraph{Further Related Work}

Feder~\cite{Feder1989new} showed that finding an egalitarian stable matching is NP-hard, and provided a 2-approximation algorithm by relating egalitarian stable matchings to the minimum vertex cover problem. Independently, Gusfield and Pitt~\cite{Gusfield1992Bounded} provided a 2-approximation for egalitarian stable matchings via a weighted minimum cost 2-SAT formulation. Teo and Sethuraman~\cite{Teo1998Geometry} provided an LP formulation of SR along with a rounding algorithm that gives a 2-approximation for optimal stable matchings when the the cost functions for each agent satisfy a ``U-shaped'' condition (which is satisfied for egalitarian costs). Cseh, Irving, and Manlove~\cite{Cseh2019Stable} showed that finding egalitarian stable matchings remains NP hard even when agents have bounded-length preference lists and provided improved approximation algorithms for maximum list lengths of at most 5. Simola and Manlove~\cite{Simola2021Profile} showed NP-hardness of finding rank-maximal and generous stable matchings with bounded preference lists as well.

Chen, Ding, Hu and Zang~\cite{Chen2012Maximum-Weight} characterized SR instances in for which the fraction stable matching polytope is integral as precisely those instances for which Phase~1 of Irving's algorithm (see Appendix~\ref{sec:Appendix2}) yields a bipartite graph. For such instances, Chen et al.\ showed that optimal stable matchings can be found in polynomial time using the same machinery as for finding optimal stable marriages in SM. This family of SR instances will have minimum crossing distance 0, thus our main result for instances with bounded minimum crossing number generalizes the algorithmic result of~\cite{Chen2012Maximum-Weight}. Combined with the characterization of~\cite{Chen2012Maximum-Weight}, our result implicitly identifies a family of SR instances whose matching polytope is not integral, but for which optimal stable matchings can be found in polynomial time.

As many variants of SM and SR are known to be NP-hard, parameterized complexity is becoming an increasingly popular tool to understand the complexity landscape of these problems~\cite{Chen2025Frontiers}. Most of the literature in this area focuses on SM variants, with relatively less work on SR. The work that is most closely related to ours is Chen et al.~\cite{Chen2018How}, who give an FPT algorithm for the egalitarian stable matching problem in SR, parameterized by the egalitarian cost of a matching.  As we have already noted, the egalitarian stable matching problem is a special case of the optimal stable matching problem where $C(a,b)$ is the rank of $b$ in $a$'s preference list.   It is unlikely that there is a relationship between the parameter we used, the minimum crossing distance, and their parameter, the egalitarian cost of a matching. Indeed, the egalitarian cost could be $\Omega(n^2)$ for SR instances with minimum crossing distance $0$.


For SM, several works have addressed the relationship between the combinatorial structure of rotation posets and the associated algorithmic complexity of SM variants. Bhanagar et al.~\cite{Bhatnagar2008Sampling} showed that certain preference restrictions still yield sufficiently general rotation posets that certain procedures for sampling stable marriages are inefficient. Subsequently, Cheng and Rosenbaum~\cite{Cheng2022Stable} fully characterized the rotation poset structure of such SM instances, both proving NP hardness and providing FPT algorithms for several SM variants. Gupta et al.~\cite{Gupta2022Treewidth} considered the complexity of SM problems parameterized by the treewidth of the instances' poset. They showed that both the balanced and sex-equal stable marriage problems are FPT when parameterized by the treewidth of the rotation poset.

%% file: SM.tex
\section{SM and the Optimal Stable Matching Problem}
\label{sec:SM}

In the {\it Stable Marriage} problem (SM), there are $n$ men and $n$ women.   Each agent has a preference list that ranks members of the opposite group in a linear ordering.  A {\it matching} partitions the $2n$ agents into $n$ pairs of men and women.  Such a matching has a {\it blocking pair} if  there is a pair of man and woman who prefer each other over their partners in the matching.  The goal is to find a {\it stable matching}, a matching with no blocking pairs. 

The seminal paper by Gale and Shapley \cite{Gale1962College} in 1962  proved that {\it every} SM instance $I$ has a stable matching.  They presented an algorithm that computes one in $O(n^2)$ time.  Interestingly, even though $I$ may have many stable matchings,  Gale and Shapley's algorithm can only output two types  -- the  {\it man-optimal/woman-pessimal} and the {\it woman-optimal/man-pessimal} stable matchings.   In the man-optimal stable matching, $\mu_M$, every man is matched to his best stable partner and, simultaneously, every woman is matched to her worst stable partner.  In the woman-optimal stable matching, $\mu_W$,  the opposite is true.  The extremeness of both of these stable matchings motivates the problem of finding {\it fair} stable matchings, those that treat both groups, and more broadly all agents,  equally.


In 1976, Knuth described the meta-structure of stable matchings that was communicated to him by Conway \cite{Knuth1976MariagesStables, Knuth1997Stable}. 
Let $M(I)$ denote the set of all stable matchings of $I$.  For $\mu \in M(I)$ and agent $a$, let $\mu(a)$ denote the partner of $a$ in $\mu$.  Define a partial ordering on $M(I)$ as follows:  for $\mu, \mu' \in M(I)$, let $\mu \preceq \mu'$ if,  for each man $m$,  $\mu(m) = \mu'(m)$ or $m$ prefers $\mu(m)$ to $\mu'(m)$.  It turns out that $\preceq$  induces a {\it distributive lattice} $\mc{M}(I)$.  
The minimum element of $\mc{M}(I)$ is $\mu_M$ while its maximum element is $\mu_W$.  As one traces a path from the minimum to the maximum element, the stable matchings gradually transform from $\mu_M$ to $\mu_W$; the partners of the men progressively worsen while the partners of the women progressively improve.  Searching the interior of $\mc{M}(I)$ directly for a fair stable matching, however, is not feasible because $I$ can have an exponential number of stable matchings (see \cite{KarlinGW18} and references therein).

Let $\mathcal{P}$ be a poset whose element set is $P$. A subset $S \subseteq P$ is {\it closed} if for every element $s \in S$, all the predecessors of $s$ are also in $S$. 
Let $C(\mc{P})$ contain all the closed subsets of $\mc{P}$.  It is easy to check that when the elements of $C(\mc{P})$ are ordered according to the subset relation, the result is a distributive lattice with $\emptyset$ as its minimum element and $P$ as its maximum element.
 A classic result by  Birkhoff \cite{Birkhoff1937Rings} in the late 1930's states the following:

\begin{theorem}[\cite{Birkhoff1937Rings}]
For every (finite) distributive lattice $\mc{D}$, there is a poset $\mc{P_D}$ whose distributive lattice of closed subsets, $(C(\mc{P_D}), \subseteq)$,   is isomorphic to $\mc{D}$.  
\label{thmBirkhoff}
\end{theorem}

We shall refer to  $\mc{P_D}$ as the {\it poset that encodes the elements of $\mc{D}$} because every element of $\mc{D}$ corresponds to a closed subset of  $\mc{P_D}$ and vice versa. In Birkhoff's proof, $\mc{P_D}$ is derived from $\mc{D}$ itself.  It is the subposet induced by the {\it join-irreducible} elements of $\mc{D}$, where the join-irreducible elements are the ones that have an in-degree  of $1$ in the Hasse diagram\footnote{We shall view the Hasse diagram of a poset as a directed graph where every edge is directed from a lower vertex to a higher one.} of $\mc{D}$.  

For stable matchings, however, Irving and Leather  \cite{Irving1986Complexity} proved in the mid-1980's that the poset that encodes the stable matchings of $I$ can be derived {\it directly} from the preference lists of $I$.  They called it the {\it rotation poset} of $I$.  The poset contains {\it rotations},  which are the minimal differences between two stable matchings. A {\it rotation} $\rho = (m_0, w_0), (m_1, w_1), \hdots, $ $(m_{r-1}, w_{r-1})$ is a cyclic list of man-woman pairs so that when $\rho$ is {\it exposed} in a stable matching $\mu$, each $(m_i, w_i) \in \mu$. When $\rho$ is {\it eliminated} from $\mu$, the result is another stable matching  
\begin{eqnarray}
 \mu/\rho & = & \mu - \{(m_i, w_i), i = 0, 1, \hdots, r-1\} \cup \{(m_i, w_{i+1}), i = 0, 1, \hdots, r-1\} 
 \label{eqn1}
 \end{eqnarray}
 where addition in the subscript is modulo $r$.  It can be shown that every stable matching that is not $\mu_W$ has an exposed rotation.  Thus, starting from $\mu_M$, 
rotations can be eliminated repeatedly until $\mu_W$ is reached.  This process reveals {\it all} the rotations of $I$ exactly once.  It is also the case that starting at $\mu_M$,  if $\mu$ can be obtained by eliminating the rotations in set $S$,  then $S$ is unique to $\mu$. That is, if $\mu$ can be obtained from $\mu_M$ by eliminating the rotations in $S'$, then $S' = S$.

Let $R(I)$ contain all the rotations of $I$.  Again, define a partial ordering on $R(I)$ as follows:   for $\rho$ and $\sigma \in R(I)$, let $\rho \leq \sigma$ if $\rho$ is eliminated before $\sigma$ in every sequence of eliminations from $\mu_M$ to $\mu_W$.  The resulting poset, $\mc{R}(I)$, is the {\it rotation poset of $I$}.  Here are the key results:

\begin{theorem}[\cite{Irving1986Complexity, Gusfield1989Stable}]
Let $I$ be an SM instance  and let $\mc{R}(I) = (R(I), \preceq)$ be its rotation poset.  The distributive lattice of stable matchings of $I$, $\mc{M}(I)$,  is isomorphic to the distributive lattice of the closed subsets of $\mc{R}(I)$,  $(C(\mc{R}(I)), \subseteq)$.   
 In particular,  a stable matching $\mu$ of $I$ corresponds to the closed subset $S_\mu$ of $\mc{R}(I)$ iff starting from $\mu_M$,  the rotations in $S_\mu$ are exactly the rotations that have to be eliminated to obtain $\mu$. 
\label{mainthm1}
\end{theorem}

Theorem~\ref{mainthm1} provides  another way of understanding how the stable matchings of $I$ are related in $\mc{M}(I)$.  
The closed subsets of $\mc{R}(I)$ that corresponds to $\mu_M$ and $\mu_W$  are $\emptyset$ and $R(I)$ respectively.  For  $\mu, \mu' \in M(I)$,  
$\mu \preceq \mu'$ if and only if $S_{\mu} \subseteq S_{\mu'}$.   Equally significant,  $\mc{R}(I)$ can be constructed from the preference lists of $I$ efficiently. 


\begin{theorem}[See discussion in Section~3.3 of \cite{Gusfield1989Stable}]
Let $I$ be an SM instance consisting of $n$ men and $n$ women.  Then $I$ has $O(n^2)$ rotations.  Furthermore, a directed acyclic graph that represents\footnote{More precisely, the transitive closure of this directed acyclic graph contains all the relations in $\mc{R}(I)$.}   
 $\mc{R}(I)$ can be constructed in $O(n^2)$ time.\label{mainthm2}
\end{theorem}

Theorems \ref{mainthm1} and \ref{mainthm2} imply that the rotation poset of $I$ is a polynomial-sized structure that encodes all the stable matchings of $I$.   It can be used to explore $\mc{M}(I)$ without actually constructing $\mc{M}(I)$.  This is exactly what Irving, Leather and Gusfield \cite{Irving1987efficient} did to find an {\it egalitarian stable matching} and, more generally, an {\it optimal stable matching} of $I$. 


\noindent{\bf Optimal Stable Matchings.} For SM instance $I$,  let $C(a,b)$ indicate how much agent $a$ likes agent $b$, a member of the opposite group.  The cost function $C$ follows the rule that for $b \neq b'$, $C(a,b) \neq C(a, b')$ (i.e., no ties are allowed) and  $C(a,b) < C(a, b')$ if and only if $a$ prefers $b$ to $b'$.  
The {\it cost of a stable matching $\mu$} is 
$$ C(\mu) = \sum_{(m,w) \in \mu} C(m,w) + C(w,m).$$
The value of $C(\mu)$ captures the dissatisfaction of the agents as a whole with respect to their matches in $\mu$. 
Given $I$ and $C$,  the objective is to find an {\it optimal stable matching}, a stable matching whose cost is the least among all the stable matchings of $I$.       When $C(a,b)$ is the {\it rank} of $b$ in $a$'s preference list (i.e., the number of agents $a$ prefers to $b$),  an optimal stable matching of $I$ is referred to as an {\it egalitarian stable matching} of $I$.  It is arguably a fair stable matching because it treats all agents the same way.    

Let $\rho = (m_0, w_0), (m_1, w_1), \hdots, (m_{r-1}, w_{r-1})$ be a rotation of $I$.  Define the {\it cost of $\rho$} with respect to $C$ as  $$C(\rho) = \sum_{i=0}^{r-1} [C(m_i, w_i) + C(w_i, m_i)] - \sum_{i=0}^{r-1}  [C(m_i, w_{i+1})+ C(w_{i+1}, m_i)].$$  Suppose $\rho$ is exposed in stable matching $\mu$.  Then from equation (\ref{eqn1}),  
\begin{eqnarray*}
C(\mu/\rho) & = &  C(\mu) - \sum_{i=0}^{r-1} [C(m_i, w_i) + C(w_i, m_i)] + \sum_{i=0}^{r-1}  [C(m_i, w_{i+1}) + C(w_{i+1}, m_i)] \\
& = & C(\mu) - C(\rho).
\end{eqnarray*}
Combining this observation with Theorem \ref{mainthm1}, the next result follows:

\begin{theorem}[\cite{Irving1987efficient}]
Let $\mu$ be a stable matching of $I$ and let $S_{\mu}$ be the closed subset in $\mc{R}(I)$ that it corresponds to.  Then 
$C(\mu) = C(\mu_M) - \sum_{\rho \in S_{\mu}} C(\rho)$.
\label{thmEgal}
\end{theorem}

For any closed subset $S$ of $\mc{R}(I)$, let the {\it cost of $S$} be $C(S) = \sum_{\rho \in S} C(\rho)$.  In the above theorem, $C(\mu_M)$ is fixed.  Thus, we can reduce the problem of finding an optimal stable matching of $I$ to  finding a closed subset of $\mc{R}(I)$ with {\it maximum cost}. The latter has been studied by other researchers (e.g. \cite{Picard1976MaximumClosure}, \cite{Hochbaum2001new}, etc.) and can be solved by computing a minimum cut of a flow network based on $\mc{R}(I)$.

\begin{theorem}[\cite{Irving1987efficient}]
Given an SM instance $I$ with $n$ men and $n$ women, an egalitarian stable matching of $I$ can be computed in $O(n^4)$ while an optimal stable matching of $I$ can be computed in $O(n^4 \log n)$ time.  
\label{corEgal}
\end{theorem}

The slight improvement in the running time for the egalitarian stable matching is due to the fact the minimum cut has a nice bound and the edge weights on the flow network are integral.  In \cite{Feder1989new}, Feder cites an $O(n^3 \log n)$ algorithm for computing an egalitarian stable matching.

%% file: SR.tex
\section{SR and the Optimal Stable Matching Problem}
\label{sec:SR}

The \emph{Stable Roommates} problem (SR) is the non-bipartite version of SM.  There are $2n$ agents and each one has a preference list that orders all the other  agents linearly.  The goal again is to find a stable matching.  Gale and Shapley \cite{Gale1962College} described SR in their original paper and pointed out that, unlike SM, some SR instances are {\it unsolvable} in that they have no stable matchings.  In 1985,  Irving \cite{Irving1985efficient} finally solved SR, presenting an $O(n^2)$ algorithm that outputs a stable matching for the  {\it solvable} instances  and reports none for the unsolvable ones.

From here on, we shall only consider SR instances that are solvable.  Let $I$ be one of them.  Irving's algorithm has two phases.  In phase 1, the preference lists of $I$ are transformed into the {\it phase-1 table} $T_0$.  In phase 2,  starting from $T_0$, exposed {\it rotations} are eliminated over and over again until none exists.  This process transforms $T_0$ into smaller and smaller tables; the final table is equivalent to a stable matching $\mu$.  When $Z$ is the set of rotations eliminated from $T_0$ to obtain $\mu$, we write $\mu = T_0/Z$.

We shall describe Irving's algorithm in more detail and define terms precisely  in Appendix~\ref{sec:Appendix2}. One important fact to note though is that Irving's algorithm does not always output the same stable matching.  The set of rotations eliminated by the algorithm determines the outcome.  That is,  suppose  $\mu = T_0/Z$ and $\mu' = T_0/Z'$.  Then $\mu = \mu'$ if and only if $Z = Z'$.


For an SR instance $I$,  let $R(I)$ contain all the exposed rotations that can be eliminated during an execution of phase 2 of Irving's algorithm.  Like SM, a {\it rotation} $\rho = (x_0, y_0), (x_1, y_1), \hdots (x_{r-1}, y_{r-1})$ in this context is a cyclic sequence of ordered pairs of agents that satisfies certain properties.  If $\ol{\rho} = (y_1, x_0), (y_2, x_1), \hdots, (y_0, x_{r-1})$ is also in $R(I)$, then $\rho$ is a {\it non-singular} rotation; otherwise, $\rho$ is a {\it singular} rotation.  If $\rho$ is non-singular, then   $\ol{\rho}$ is referred to as the {\it dual} of $\rho$.  Notice that  $\ol{\rho}$ is also non-singular and its dual is $\rho$ so $\{\rho,  \ol{\rho} \}$ is a {\it dual pair} of non-singular rotations of $I$.

\begin{theorem}[{\cite[Theorem 4.3.1]{Gusfield1989Stable}}]
	Let $\mu$ be a stable matching of $I$.  If $\mu = T_0/Z$,  then $Z$ contains every singular rotation and exactly one of each dual pair of non-singular rotations of $I$.
	\label{thmZcontent}
\end{theorem}


Just like in the SM setting,  a poset can be created from $R(I)$.  For $\sigma, \rho \in R(I)$,  let $\sigma \leq \rho$ if for every sequence of eliminations in the phase 2 of Irving's algorithm in which $\rho$ appears, $\sigma$ appears before $\rho$.  The poset $\mc{R}(I) = (R(I), \leq)$ is the {\it rotation poset of $I$}.  Here are some properties about the ordering $\leq$:

\begin{lemma} (Lemma 4.3.7 in \cite{Gusfield1989Stable}.)
	Let $\rho$ and $\sigma$ be non-singular rotations and $\tau$ be a singular rotation of $I$.  Then

	\noindent (i) $\rho$ and $\ol{\rho}$ are incomparable elements in $\mc{R}(I)$;

	\noindent (ii) $\sigma \leq \rho$ if and only if $\ol{\rho} \leq \ol{\sigma}$;

	\noindent (iii) any predecessor of $\tau$ in $\mc{R}(I)$ is a singular rotation.
	\label{lemmaRotations}
\end{lemma}

Let $R'(I)$ consist of all the non-singular rotations of $I$.  Let $\mc{R}'(I)$ be the subposet  induced by $R'(I)$; it is referred to as the {\it reduced rotation poset of $I$}.  A subset $S \subseteq R'(I)$ is {\it complete} if it contains exactly one element from each dual pair of rotations of $I$.  It is {\it closed} if whenever $\rho \in S$, all the predecessors of $\rho$ are also in $S$.   Here are the counterparts of Theorems \ref{mainthm1} and \ref{mainthm2}:

\begin{theorem}[\cite{Gusfield1988Structure, Irving1986Stable}]
	Let $I$ be an SR instance and let $\mc{R}'(I) = (R'(I), \leq)$ be its reduced rotation poset.  There is a one-to-one correspondence between the stable matchings of $I$ and the complete closed subsets of $\mc{R}'(I)$.  In particular, a stable matching $\mu$ of $I$ corresponds to the complete closed subset $S_\mu$ of $\mc{R}'(I)$ if and only if starting from $T_0$,  all the singular rotations of $I$ and all the rotations in $S_\mu$ are exactly the rotations that have to be eliminated to obtain $\mu$.
	\label{mainthm3}
\end{theorem}

\begin{theorem}[Lemma 6.4 in \cite{Gusfield1988Structure}, Section 4.4.1 in \cite{Gusfield1989Stable}]
	Let $I$ be an SR instance with $2n$ agents.  Then $I$ has $O(n^2)$ rotations.  Furthermore, there is a directed acyclic graph that represents $\mc{R}'(I)$ that has $O(n^2)$ vertices and $O(n^2)$ edges, and it can be constructed in $O(n^3 \log n)$ time.
	\label{thm:sr-poset}
\end{theorem}

\begin{corollary}
	Let $I$ be an SR instance with $2n$ agents and let  $\mc{R}'(I)$ its reduced rotation poset.  Given a complete closed subset $S$ of $\mc{R}'(I)$,  the stable matching of $I$ that corresponds to $S$ can be obtained in $O(n^2)$  time.
	\label{corFindMatching}
\end{corollary}
\begin{proof}
	See Appendix~\ref{sec:Appendix2}.
\end{proof}

Thus, like the SM setting,  there is a polynomial-sized structure that encodes the stable matchings of an SR instance I.  What is missing from this discussion though is the meta-structure of the stable matchings of $I$.  Indeed, it was not until the work of Cheng and Lin~\cite{Cheng2011Stable} in 2011 that this meta-structure was fully understood.   They took a combinatorial approach and viewed $\mc{R}'(I)$ in general terms, calling it a {\it mirror poset}.

\begin{definition}
	In a mirror poset $\mathcal{P} = (P, \leq)$,  the element set $P$ is partitioned into dual pairs of elements, where the dual of $\rho \in P$ is denoted as $\ol{\rho}$, such that

	\noindent (i) for each $\rho \in P$,  $\rho$ and its dual $\ol{\rho}$ are incomparable and

	\noindent (ii) for any $\sigma, \rho \in P$, $\sigma \leq \rho$ if and only if $\ol {\rho} \leq \ol{\sigma}$.
\end{definition}

A subset $S \subseteq P$ is {\it partially complete} if it contains at most one element from each dual pair  in $\mathcal{P}$; it is {\it complete} if it contains {\it exactly} one element from each pair in $\mathcal{P}$.  It is {\it closed} whenever $\rho \in S$, every predecessor of $\rho$ is also in $S$.   

\begin{definition}
	A graph $G=(V,E)$ is a median graph if for any three vertices $u, v, w$, there exists a unique vertex that lies in a shortest path from $u$ to $v$, $u$ to $w$ and $v$ to $w$.
\end{definition}

As we shall see, mirror posets are related to {median graphs}, a well-studied class of graphs that includes hypercubes, trees and
the cover graphs\footnote{If $\mc{P}$ is a poset, the {\it cover graph} of $\mc{P}$ is the undirected version of the Hasse diagram of $\mc{P}$.} of distributive lattices.  Median graphs are connected, bipartite and have at most $|V| \log_2 |V|$ edges \cite{Imrich2000Product}.

\begin{theorem}[{\cite[Corollary 1 and Lemma 4]{Cheng2011Stable}}] 
	Let $\mathcal{P}$ be a mirror poset.   Denote as $G(\mathcal{P})$ the graph whose vertices are the complete closed subsets of $\mathcal{P}$ and where any two of them are adjacent if and only if they differ by one dual element.  Then $G(\mathcal{P})$ is a median graph.   Moreover, for any two closed subsets $S$ and $S'$ of $\mathcal{P}$, their distance in $G(\mathcal{P})$ is  $d(S, S') =  |S-S'| = |S'-S|$. 
	\label{thmMirrorToMedian}
\end{theorem}

Figure \ref{figmirror} shows an example of a mirror poset and the median graph formed by its complete closed subsets. Using the correspondence described in Theorem \ref{mainthm3}, we now have a meta-structure for the stable matchings of $I$:

\begin{corollary}[\cite{Cheng2011Stable}]
	For an SR instance $I$, let $G(I)$ denote the graph whose vertices are the stable matchings of $I$ and any two stable matchings $\mu$ and $\mu'$ are adjacent if and only if $S_{\mu}$ and $S_{\mu'}$ differ by one rotation.  Then $G(I)$ is a median graph.  Moreover, for stable matchings $\mu$ and $\eta$, their distance in $G(I)$ is $|S_\mu - S_\eta| = |S_\eta - S_\mu|$.
	\label{corSRMeta}
\end{corollary}

It is interesting that median graphs, the meta-structures for SR stable matchings,  are unordered while distributive lattices, the meta-structures for SM stable matchings,  are ordered.  The difference has to do with the fact that SM instances have two sides -- the men and the women -- and these two sides are responsible for the meet and join operation and, consequently, the minimum and maximum elements of the distributive lattices.  In contrast, there are no sides to SR instances and therefore creating an ordering for their stable matchings is less natural.

\begin{figure}
	\begin{center}
		\begin{tikzpicture}
			\node (1+) at (-2,0) {$\rho_1$};
			\node (2+) at (0,0) {$\rho_2$};
			\node (3+) at (2,0) {$\rho_3$};
			\node (4+) at (-3,1) {$\rho_4$};
			\node (5+) at (-1,1) {$\rho_5$};
			\node (1-) at (2,2) {$\ol{\rho_1}$};
			\node (2-) at (0,2) {$\ol{\rho_2}$};
			\node (3-) at (-2,2) {$\ol{\rho_3}$};
			\node (4-) at (1,1) {$\ol{\rho_4}$};
			\node (5-) at (3,1) {$\ol{\rho_5}$};

			\begin{scope}[every edge/.style={->,draw=black, thick}]
				\draw (1+) edge node[above]{} (4+);
				\draw (4-) edge node[above]{} (1-);
				\draw (1+) edge node[above]{} (5+);
				\draw (5-) edge node[above]{} (1-);
				\draw (2+) edge node[above]{} (5+);
				\draw (5-) edge node[above]{} (2-);
				\draw (2+) edge node[above]{} (4-);
				\draw (4+) edge node[above]{} (2-);
				\draw (3+) edge node[above]{} (4-);
				\draw (4+) edge node[above]{} (3-);
				\draw (3+) edge node[above]{} (5-);
				\draw (5+) edge node[above]{} (3-);
			\end{scope}
		\end{tikzpicture}
		\hspace*{2em}
		\begin{tikzpicture}
			\node (1) at (-2.5,0) {$\rho_1 \rho_2 \ol{\rho_3} \rho_4 \rho_5$};
			\node (2) at (0,0) {$\rho_1 \rho_2 {\rho_3} \rho_4 \rho_5$};
			\node (3) at (2.5,0) {$\rho_1 \rho_2 {\rho_3} \ol{\rho_4} \rho_5$};
			\node (4) at (0,1.25) {$\rho_1 \rho_2 {\rho_3} \rho_4 \ol{\rho_5}$};
			\node (5) at (2.5,1.25) {$\rho_1 \rho_2 {\rho_3} \ol{\rho_4} \ol{\rho_5}$};
			\node (6) at (0,2.5) {$\rho_1 \ol{\rho_2} {\rho_3} \rho_4 \ol{\rho_5}$};
			\node (7) at (2.5,2.5) {$\ol{\rho_1} \rho_2 {\rho_3} \ol{\rho_4} \ol{\rho_5}$};

			\begin{scope}[every edge/.style={-,draw=black, thick}]
				\draw (1) edge node[above]{} (2);
				\draw (2) edge node[above]{} (3);
				\draw (2) edge node[above]{} (4);
				\draw (3) edge node[above]{} (5);
				\draw (4) edge node[above]{} (5);
				\draw (4) edge node[above]{} (6);
				\draw (5) edge node[above]{} (7);
			\end{scope}
		\end{tikzpicture}
	\end{center}
	\caption{On the left is a mirror poset $\mc{P}$ with five dual pairs.    On the right is the median graph $G(\mc{P})$ formed by its seven complete closed subsets. Two subsets are adjacent in the median graph if and only if they differ by one element.}
	\Description{A side-by-side diagram. The left panel is a directed mirror poset with five paired rotations, and the right panel is the seven-vertex median graph formed by its complete closed subsets.}
	\label{figmirror}
\end{figure}
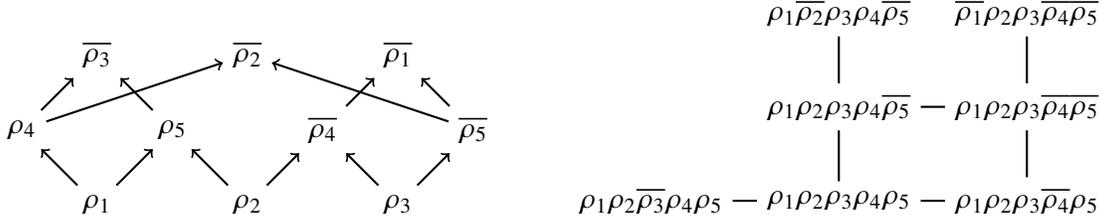




\paragraph{Optimal Stable Matchings} We now consider the optimal stable matching problem for SR instance $I$.  As before, there is a cost function $C$ so that for distinct agents $a, b, b'$,  $C(a,b) \neq C(a, b')$  and $C(a, b) < C(a, b')$ if and only if $a$ prefers $b$ to $b'$.  The {\it cost} of a stable matching $\mu$ is $$C(\mu) = \sum_{\{a, b\} \in \mu} C(a,b) + C(b,a).$$
A stable matching of $I$ with the least cost is an {\it optimal stable matching}.  When $C(a,b)$ is the rank of $b$ in $a$'s preference list, then an optimal stable matching of $I$ is also called an {\it egalitarian stable matching} of $I$.
The structural similarities between SM stable matchings and SR stable matchings suggest that the approach for finding an optimal stable matching in the SM setting might work in the SR setting.   To do so, however, we need a result similar to equation (\ref{eqn1}) that describes how two adjacent stable matchings in $\mc{M}(I)$ are related.

\begin{lemma}
	Let $\mu$ and $\mu'$ be two stable matchings of SR instance $I$ that are adjacent in $G(I)$.
	Let $\rho = (x_0, y_0), (x_1, y_1), \hdots,$ $(x_{r-1}, y_{r-1})$ be the rotation so that $S_{\mu'} - S_\mu =  \{\rho\}$.  Then
	\begin{center}
		$\mu' = \mu - \{ \{x_0, y_0\}, \{x_1, y_1\}, \hdots, \{x_{r-1}, y_{r-1}\} \} \cup \{ \{x_0, y_1\}, \{x_1, y_2\}, \hdots, \{x_{r-1}, y_{0}\} \}.$
	\end{center}
	\label{lemmaDiff}
\end{lemma}
\begin{proof}  See Appendix~\ref{sec:Appendix2}.
\end{proof}

Proving the lemma above is requires work because the rotations of $I$ are exposed with respect to tables rather than stable matchings.  We work out the details in Appendix~\ref{sec:Appendix2}.

For rotation $\rho  = (x_0, y_0), (x_1, y_1), \hdots, (x_{r-1}, y_{r-1})$, let us define its cost as
\[C(\rho) = \sum_{i=0}^{r-1} [C(x_i, y_i)+ C(y_i, x_i)] - \sum_{i=0}^{r-1} [C(x_i, y_{i+1}) + C(y_{i+1}, x_i)]\]
where the addition in the subscript is modulo $r$.

\begin{theorem}
	Let $\mu$ and $\mu'$ be two stable matchings of SR instance $I$.
	Then
	$$C(\mu') = C(\mu) - \!\!\! \sum_{\rho \in S_{\mu'} - S_\mu} \!\!\!C(\rho).$$
	\label{thmSREgal}
\end{theorem}
\begin{proof}
	From Corollary \ref{corSRMeta},  the distance between $\mu$ and $\mu'$ in $G(I)$ is $|S_{\mu'} - S_\mu|$. Assume $|S_{\mu'} - S_\mu| = \ell$.  Then there is a path $\Psi$ from $\mu$ to $\mu'$ in $G(I)$ so that $\Psi = (\mu_0, \mu_1, \hdots, \mu_\ell)$ with $\mu_0 = \mu$ and $\mu_\ell = \mu'$.  For simplicity, let $S_i = S_{\mu_i}$ for $i = 0, \hdots, \ell$.

	Assume $\rho \in   S_{\mu'} - S_\mu$.   Then $\ol{\rho} \in S_\mu$ because $S_\mu$ is a complete subset of $\mc{R}'(I)$.   It must also be the case that for some $i$,  $S_i - S_{i-1} = \{\rho\}$.  Otherwise,  $\ol{\rho} \in S_0, S_1, \hdots, S_\ell$.  Since $S_\ell = S_{\mu'}$ and $\rho \in   S_{\mu'}$,   both $\rho$ and $\ol{\rho}$  are in   $S_{\mu'}$,  a contradiction.  Hence, {\it every} element in $S_{\mu'} - S_\mu$ can be found in $\cup_{i=1}^\ell (S_i - S_{i-1})$.  That is,  $|\cup_{i=1}^\ell (S_i - S_{i-1})| \ge \ell$.  But  $\cup_{i=1}^\ell (S_i - S_{i-1})$ cannot have more than $\ell$ elements since $|S_i-S_{i-1}| = 1$ for each $i$.  Thus,  $\cup_{i=1}^\ell (S_i - S_{i-1}) = S_{\mu'} - S_\mu$.

	Let $S_i - S_{i-1} = \rho_i$ for $i = 1, \hdots, \ell$.  From our discussion, $S_{\mu'} - S_\mu = \{\rho_1, \rho_2, \hdots, \rho_\ell \}$.
	Lemma \ref{lemmaDiff} implies that $C(\mu_i) = C(\mu_{i-1}) - C(\rho_i)$  for $i = 1, \hdots, \ell$.  Thus,
	\begin{eqnarray*}
		C(\mu_\ell) & = & C(\mu_{\ell-1}) - C(\rho_\ell) \\
		& = & [C(\mu_{\ell -2}) - C(\rho_{\ell -1})] - C(\rho_\ell) \\
		& = &  [C(\mu_{\ell -3}) - C(\rho_{\ell -2})] - C(\rho_{\ell -1}) - C(\rho_\ell) \\
		& \vdots & \\
		& = & C(\mu_0) - C(\rho_1) - C(\rho_2) - \cdots -  C(\rho_\ell) \\
		& = & C(\mu_0) - \!\!\! \sum_{\rho \in S_{\mu'} - S_\mu} \!\!\! C(\rho).
	\end{eqnarray*}
\end{proof}

Notice that Theorem \ref{thmSREgal} is quite similar to Theorem \ref{thmEgal}.  If $\mu$ is fixed, then finding an optimal stable matching of $I$ is equivalent to finding a complete closed subset $S$ of $\mc{R}'(I)$ that maximizes  $\sum_{\rho \in S - S_\mu} C(\rho)$.  Unfortunately, this is not an easy problem.  Feder \cite{Feder1989new, Feder1992journal} has shown that computing an egalitarian stable matching in the SR setting is NP-hard.
Nonetheless, this approach to finding an optimal stable matching can still be fruitful.


%% file: MedianSemilattices.tex
\section{Median Semilattices and Local Optima}
\label{sec:MedianSemilattices}

We have already noted that the meta-structure of an SR instance's stable matchings is unordered.  But it is still possible to impose an ordering on it because a median graph can be transformed into an ordered structure called a {\it median semilattice}.

\begin{definition}
A {\it median semilattice} $\mathcal{L}$ is a meet semilattice (i.e., the greatest lower bound of any two elements always exists) such that (i) for any element $v$, the subposet induced by $v$ and its  predecessors is a distributive lattice and (ii) any three elements have an upper bound whenever any two of them have an upper bound.
\label{defmedsemi}
\end{definition}


The next proposition states when a median semilattice is a distributive lattice.

\begin{proposition}
A median semilattice $\mathcal{L}$ is a distributive lattice if and only if it has a maximum element.
\label{propmax}
\end{proposition}

\begin{proof}
Let $\hat{0}$ be the minimum element of $\mc{L}$.  Suppose $\mathcal{L}$ has a maximum element $\hat{1}$.  From the definition of median semilattices, the subposet induced $\hat{1}$ and its predecessors  is a distributive lattice.  But this subposet contains all the elements of $\mathcal{L}$ so $\mathcal{L}$ itself is a distributive lattice.
On the other hand, if $\mathcal{L}$ is a distributive lattice, it must have a maximum element by definition.
\end{proof}

Let $G$ be a connected graph and $u$ be a vertex of $G$.  The {\it canonical order $\sq_u$} induces a partial  ordering on the vertices $G$ as follows:  for vertices $v$ and $w$, $v \sq_u w$ if and only if $d(u, w) = d(u, v) + d(v, w)$.  That is, $v \sq_u w$ if and only if $v$ lies on a shortest path from $u$ to $w$.   A result by Avann \cite{Avann1961metric} connects median graphs to median semilattices.  Cheng and Lin \cite{Cheng2011Stable} clarifies the cover graph of the median semilattices.

\begin{theorem}
(\cite{Avann1961metric,Cheng2011Stable}  From median graphs to median semilattices.)  Let $G=(V,E)$ be a median graph and $u$ be a vertex of $G$.  The poset $(V, \sq_u)$ is a median semilattice whose minimum element is $u$.   Moreover, the cover graph of $(V, \sq_u)$ is $G$.
\label{thmMedianToSemi}
\end{theorem}

An implication of the theorem is that if we take a median graph $G$, root it at one of its vertices $u$ and direct all the edges of $G$ away from $u$, the result is the Hasse diagram of the median semilattice $(V, \sq_u)$. For this reason, we  refer to $(V, \sq_u)$ as  the median semilattice obtained {\it by rooting $G$ at $u$}.  Note that the same result holds if $u$ is replaced by another vertex $v$ of $G$.  Thus,  a median graph $G$ can be transformed into {\it one or more} median semilattices.  Figures \ref{figsemis} and \ref{figsemis2} have some examples.

Suppose $\mathcal{P} = (P, \leq)$ is a mirror poset.  Let $G(\mathcal{P})$ be the median graph formed by the complete closed subsets of $\mc{P}$ as described in Theorem \ref{thmMirrorToMedian}.    Root $G(\mathcal{P})$ at a complete closed subset $W$.   According to our discussion above,  the result is a  median semilattice,  which we denote as $\mc{L}(\mc{P}, W)$ with ordering relation  $\sq_W$.   To better understand $\mc{L}(\mc{P}, W)$, we now introduce the notion of {\it orienting} $\mc{P}$.  One way of thinking about $\mc{L}(\mc{P}, W)$ is that we are viewing $G(\mathcal{P})$ from the perspective of $W$.  A corresponding action on $\mc{P}$ is to orient $\mc{P}$ at $W$.

\begin{definition}
    Let $\mathcal{P} = (P, \leq)$ be a mirror poset.   An {\it orientation of $\mc{P}$} is a partitioning of $P$ into $P^- \cup P^+$ where $P^-$ is a complete closed subset of $\mc{P}$ and $P^+ = P - P^-$.  We denote this orientation as $\mc{P}= (P^- \cup P^+, \leq)$ and say that $\mc{P}$ is {\it oriented} at $P^-$.  We also refer to $P^-$ as the {\it base} of the orientation.  The {\it crossing edges} of $\mc{P}= (P^- \cup P^+, \leq)$ are the edges between $P^-$ and $P^+$ in the Hasse diagram of $\mc{P}$.  When the orientation of $\mc{P}$ at $P^-$ is already specified, we  denote the median semilattice $\mc{L}(\mc{P}, P^-)$ simply as $\mc{L}(\mc{P})$ and the ordering $\sq_{P^-}$ as $\sq$.
\end{definition}

In our discussion,  the elements in $P^-$ and $P^+$ are called the {\it negative} and {\it positive} elements of $\mc{P}$ respectively.  Since $P^-$ is a complete subset of $\mc{P}$, every dual pair  has an element in $P^-$ and another in $P^+$.  Thus,
for each dual pair $\{\rho, \overline{\rho}\}$, we rename the elements as $\{\rho^-, \rho^+\}$ where $\rho^- \in P^-$ and $\rho^+ \in P^+$.  Additionally, because $P^-$ is a closed subset of $\mc{P}$, none of its elements has a predecessor in $P^+$, so any crossing edges of $\mc{P}= (P^- \cup P^+, \leq)$ must be directed from $P^-$ to $P^+$.  For this reason, when illustrating an orientation of $\mathcal{P}$,  the elements of $P^-$ appear at the bottom while those of $P^+$  at the top.  See Figures \ref{figsemis} and \ref{figsemis2} again for some examples.

\begin{figure}
\begin{center}
\begin{tikzpicture}[scale=0.9]
\node (1) at (-2.25,1.25) {$\rho_1 \rho_2 \ol{\rho_3} \rho_4 \rho_5$};
\node (2) at (0,0) {$\rho_1 \rho_2 {\rho_3} \rho_4 \rho_5$};
\node (3) at (2.25,1.25) {$\rho_1 \rho_2 {\rho_3} \ol{\rho_4} \rho_5$};
\node (4) at (0,1.25) {$\rho_1 \rho_2 {\rho_3} \rho_4 \ol{\rho_5}$};
\node (5) at (2.25,2.5) {$\rho_1 \rho_2 {\rho_3} \ol{\rho_4} \ol{\rho_5}$};
\node (6) at (0,2.5) {$\rho_1 \ol{\rho_2} {\rho_3} \rho_4 \ol{\rho_5}$};
\node (7) at (2.25,3.75) {$\ol{\rho_1} \rho_2 {\rho_3} \ol{\rho_4} \ol{\rho_5}$};

\begin{scope}[every edge/.style={->,draw=black, thick}]
\draw (2) edge node[above]{} (1);
\draw (2) edge node[above]{} (3);
\draw (2) edge node[above]{} (4);
\draw (3) edge node[above]{} (5);
\draw (4) edge node[above]{} (5);
\draw (4) edge node[above]{} (6);
\draw (5) edge node[above]{} (7);
\end{scope}
  \draw[->, black, thick, dashed] (0,-1) -- (0,-2);
   \draw[->, black, thick, dashed] (4,-4.75) -- (5,-4.75);
\node (1+) at (-2,-6.5) {$\rho_1$};
\node (2+) at (0,-6.5) {$\rho_2$};
\node (3+) at (2,-6.5) {$\rho_3$};
\node (4+) at (-1,-5.5) {$\rho_4$};
\node (5+) at (1,-5.5) {$\rho_5$};
\node (1-) at (2,-3) {$\ol{\rho_1}$};
\node (2-) at (0,-3) {$\ol{\rho_2}$};
\node (3-) at (-2,-3) {$\ol{\rho_3}$};
\node (4-) at (-1,-4) {$\ol{\rho_4}$};
\node (5-) at (1,-4) {$\ol{\rho_5}$};

\begin{scope}[every edge/.style={->,draw=black, thick}]
\draw (1+) edge node[above]{} (4+);
\draw (4-) edge node[above]{} (1-);
\draw (1+) edge node[above]{} (5+);
\draw (5-) edge node[above]{} (1-);
\draw (2+) edge node[above]{} (5+);
\draw (5-) edge node[above]{} (2-);
\draw (2+) edge node[above]{} (4-);
\draw (4+) edge node[above]{} (2-);
\draw (3+) edge node[above]{} (4-);
\draw (4+) edge node[above]{} (3-);
\draw (3+) edge node[above]{} (5-);
\draw (5+) edge node[above]{} (3-);
\end{scope}
\node (1+) at (7,-6.5) {$\rho^-_1$};
\node (2+) at (9,-6.5) {$\rho^-_2$};
\node (3+) at (11,-6.5) {$\rho^-_3$};
\node (4+) at (8,-5.5) {$\rho^-_4$};
\node (5+) at (10,-5.5) {$\rho^-_5$};
\node (1-) at (11,-3) {${\rho^+_1}$};
\node (2-) at (9,-3) {${\rho^+_2}$};
\node (3-) at (7,-3) {${\rho^+_3}$};
\node (4-) at (8,-4) {${\rho^+_4}$};
\node (5-) at (10,-4) {${\rho^+_5}$};
\draw[dashed] (7,-4.75) -- (11,-4.75);

\begin{scope}[every edge/.style={->,draw=black, thick}]
\draw (1+) edge node[above]{} (4+);
\draw (4-) edge node[above]{} (1-);
\draw (1+) edge node[above]{} (5+);
\draw (5-) edge node[above]{} (1-);
\draw (2+) edge node[above]{} (5+);
\draw (5-) edge node[above]{} (2-);
\draw (2+) edge node[above]{} (4-);
\draw (4+) edge node[above]{} (2-);
\draw (3+) edge node[above]{} (4-);
\draw (4+) edge node[above]{} (3-);
\draw (3+) edge node[above]{} (5-);
\draw (5+) edge node[above]{} (3-);
\end{scope}
\node (1) at (6,1.25) {$\rho_1^- \rho_2^- \rho_3^+ \rho_4^- \rho_5^-$};
\node (2) at (9,0) {$\rho_1^- \rho_2^- \rho_3^- \rho_4^- \rho_5^-$};
\node (3) at (12,1.25) {$\rho_1^- \rho_2^- \rho_3^- {\rho_4}^+ \rho_5^-$};
\node (4) at (9,1.25) {$\rho_1^- \rho_2^- \rho_3^- \rho_4^- {\rho_5}^+$};
\node (5) at (12, 2.5) {$\rho_1^- \rho_2^- \rho_3^- {\rho_4}^+ {\rho_5}^+$};
\node (6) at (9,2.5) {$\rho_1^- {\rho_2}^+ \rho_3^- \rho_4^- {\rho_5}^+$};
\node (7) at (12,3.75) {${\rho_1}^+ \rho_2^- \rho_3^- {\rho_4}^+ {\rho_5}^+$};

\begin{scope}[every edge/.style={->,draw=black, thick}]
\draw (2) edge node[above]{} (1);
\draw (2) edge node[above]{} (3);
\draw (2) edge node[above]{} (4);
\draw (3) edge node[above]{} (5);
\draw (4) edge node[above]{} (5);
\draw (4) edge node[above]{} (6);
\draw (5) edge node[above]{} (7);
\end{scope}
  \draw[<-, black, thick, dashed] (9,-1) -- (9,-2);

\end{tikzpicture}
\end{center}
\caption{At the top left is the median semilattice $\mc{L}(\mc{P}, W)$ formed by rooting the median graph $G(\mc{P})$ in Figure \ref{figmirror} at the complete closed subset $W = \{\rho_1, \rho_2, \rho_3, \rho_4, \rho_5\}$.   It has three maximal elements.   At the bottom left is the mirror poset $\mc{P}$ oriented at $W$; the lower half of $\mc{P}$ consists of the elements of $W$ and the upper half has the elements of $P-W$.  Each $\rho_i$ is renamed as $\rho^-_i$ and its dual as $\rho_i^+$.  The renamed poset is shown on the bottom right.   It has three pairs of crossing edges from $\mc{P}^-$ to $\mc{P}^+$: $\{(\rho_2^-, \rho_4^+),  (\rho_4^-, \rho_2^+)\}$, $\{(\rho_3^-, \rho_4^+),  (\rho_4^-, \rho_3^+)\}$ and $\{(\rho_3^-, \rho_5^+),  (\rho_5^-, \rho_3^+)\}$.   At the top right, the complete closed subsets  of $\mc{P}$ are updated to reflect the change in the naming convention. }
\Description{A four-panel diagram showing a median semilattice rooted at a complete closed subset, the corresponding oriented mirror poset, the same poset with elements renamed negative and positive, and the resulting complete closed subsets.}
\label{figsemis}
\end{figure}
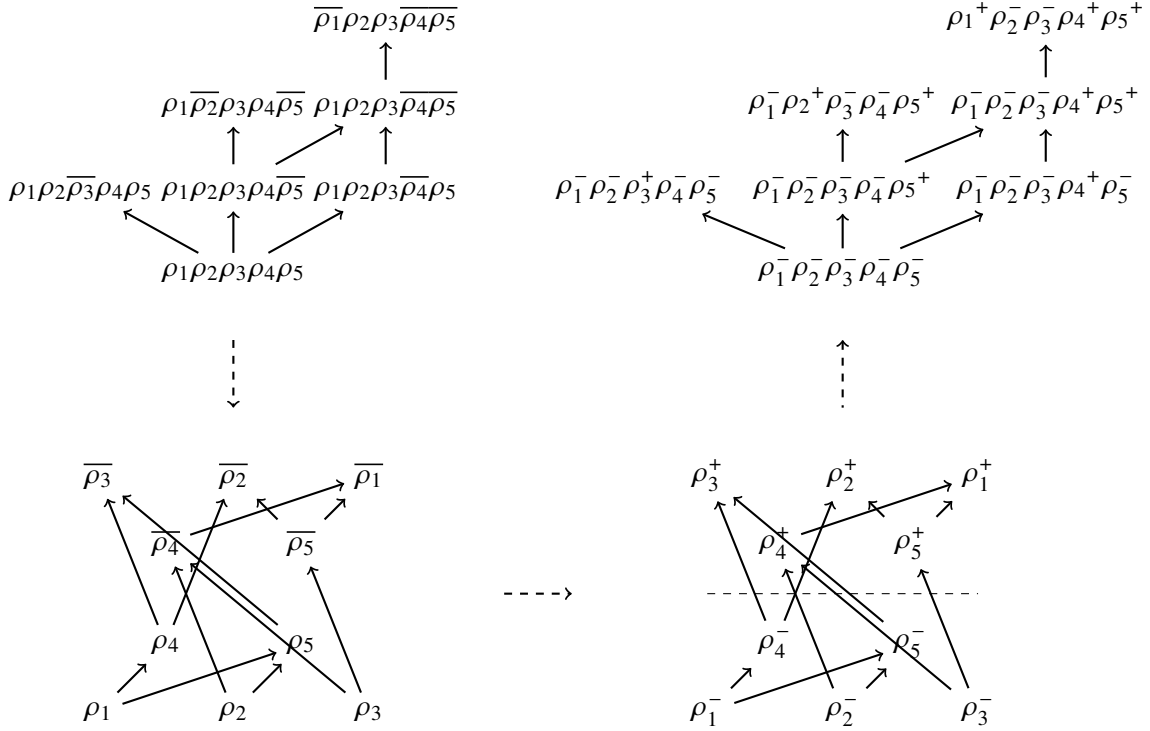

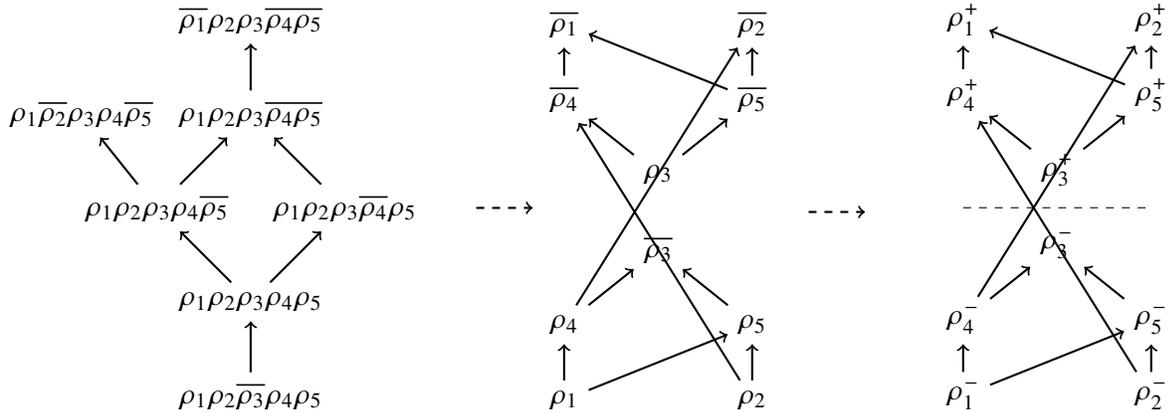
\begin{figure}
\begin{center}
\begin{tikzpicture}
\node (1) at (0, 0) {$\rho_1 \rho_2 \ol{\rho_3} \rho_4 \rho_5$};
\node (2) at (0, 1.25) {$\rho_1 \rho_2 {\rho_3} \rho_4 \rho_5$};
\node (3) at (1.25, 2.5) {$\rho_1 \rho_2 {\rho_3} \ol{\rho_4} \rho_5$};
\node (4) at (-1.25,2.5) {$\rho_1 \rho_2 {\rho_3} \rho_4 \ol{\rho_5}$};
\node (5) at (0, 3.75) {$\rho_1 \rho_2 {\rho_3} \ol{\rho_4} \ol{\rho_5}$};
\node (6) at (-2.25,3.75) {$\rho_1 \ol{\rho_2} {\rho_3} \rho_4 \ol{\rho_5}$};
\node (7) at (0,5) {$\ol{\rho_1} \rho_2 {\rho_3} \ol{\rho_4} \ol{\rho_5}$};

\begin{scope}[every edge/.style={->,draw=black, thick}]
\draw (1) edge node[above]{} (2);
\draw (2) edge node[above]{} (3);
\draw (2) edge node[above]{} (4);
\draw (3) edge node[above]{} (5);
\draw (4) edge node[above]{} (5);
\draw (4) edge node[above]{} (6);
\draw (5) edge node[above]{} (7);
\end{scope}
\draw[->, black, thick, dashed] (3,2.5) -- (3.75,2.5);
\end{tikzpicture}
\begin{tikzpicture}
\node (1+) at (0,0) {$\rho_1$};
\node (2+) at (2.5,0) {$\rho_2$};
\node (3+) at (2,3) {$\rho_3$};
\node (4+) at (0,1) {$\rho_4$};
\node (5+) at (2.5,1) {$\rho_5$};
\node (1-) at (0,5) {$\ol{\rho_1}$};
\node (2-) at (2.5,5) {$\ol{\rho_2}$};
\node (3-) at (2,2) {$\ol{\rho_3}$};
\node (4-) at (0,4) {$\ol{\rho_4}$};
\node (5-) at (2.5,4) {$\ol{\rho_5}$};

\begin{scope}[every edge/.style={->,draw=black, thick}]
\draw (1+) edge node[above]{} (4+);
\draw (4-) edge node[above]{} (1-);
\draw (1+) edge node[above]{} (5+);
\draw (5-) edge node[above]{} (1-);
\draw (2+) edge node[above]{} (5+);
\draw (5-) edge node[above]{} (2-);
\draw (2+) edge node[above]{} (4-);
\draw (4+) edge node[above]{} (2-);
\draw (3+) edge node[above]{} (4-);
\draw (4+) edge node[above]{} (3-);
\draw (3+) edge node[above]{} (5-);
\draw (5+) edge node[above]{} (3-);
\end{scope}
\draw[->, black, thick, dashed] (3.25,2.5) -- (4,2.5);
\end{tikzpicture}
\hspace*{2em}
\begin{tikzpicture}
\node (1+) at (0,0) {$\rho^-_1$};
\node (2+) at (2.5,0) {$\rho^-_2$};
\node (3+) at (2,3) {$\rho^+_3$};
\node (4+) at (0,1) {$\rho^-_4$};
\node (5+) at (2.5,1) {$\rho^-_5$};
\node (1-) at (0,5) {${\rho^+_1}$};
\node (2-) at (2.5,5) {${\rho^+_2}$};
\node (3-) at (2,2) {${\rho^-_3}$};
\node (4-) at (0,4) {${\rho^+_4}$};
\node (5-) at (2.5,4) {${\rho^+_5}$};
\draw[dashed] (0, 2.5) -- (2.5,2.5);

\begin{scope}[every edge/.style={->,draw=black, thick}]
\draw (1+) edge node[above]{} (4+);
\draw (4-) edge node[above]{} (1-);
\draw (1+) edge node[above]{} (5+);
\draw (5-) edge node[above]{} (1-);
\draw (2+) edge node[above]{} (5+);
\draw (5-) edge node[above]{} (2-);
\draw (2+) edge node[above]{} (4-);
\draw (4+) edge node[above]{} (2-);
\draw (3+) edge node[above]{} (4-);
\draw (4+) edge node[above]{} (3-);
\draw (3+) edge node[above]{} (5-);
\draw (5+) edge node[above]{} (3-);
\end{scope}
\end{tikzpicture}
\hspace*{2em}

\end{center}
\caption{This time around, the median graph $G(\mc{P})$ in Figure \ref{figmirror} is rooted at $W = \{\rho_1, \rho_2, \ol{\rho_3}, \rho_4, \rho_5\}$ to create another median semilattice. This semilattice has two maximal elements.  The mirror poset $\mathcal{P}$ is oriented at $W$ and the rotations are renamed to reflect this orientation.  The resulting poset on the right has only one pair of crossing edges: $\{(\rho_2^-, \rho_4^+), (\rho_4^-, \rho_2^+)\}$. }
\Description{A three-panel diagram showing a median semilattice rooted at a different complete closed subset, the corresponding oriented mirror poset, and the renamed poset with one pair of crossing edges.}
\label{figsemis2}
\end{figure}

Suppose $\mathcal{P} = (P^- \cup P^+, \leq)$ and $Q \subseteq P^- \cup P^+$.  Let $Q^- = Q \cap P^-$ and $Q^+ = Q \cap P^+$.  The next lemma provides a simple way of determining when two complete closed subsets of $\mc{P}$ are comparable in $\mc{L}(\mc{P})$.


\begin{lemma}
Let $\mathcal{P} = (P^- \cup P^+, \leq)$ and let $S$ and $T$ be two complete closed subsets of $\mathcal{P}$.  Then,   $S \sq T$   in $\mathcal{L}(\mathcal{P})$ if and only if $S^+ \subseteq T^+$.
\label{lemmasubset}
\end{lemma}

\begin{proof}
Consider the complete closed subsets $S$ and $T$.   By definition,  $\mc{L}(\mc{P})$ is rooted at $P^-$
 so $S \sq T$  if and only if $ S \sq_{P^-} T$.  Using the definition of $\sq_{P^-}$ and the fact that $P^-$ contains only negative elements, we have
\begin{eqnarray*}
 S \sq_{P^-} T & \iff & d(P^-, T) = d(P^-, S) + d(S, T) \\
  & \iff & |T^+| = |S^+| + d(S,T) \\
  & \iff & d(S,T) = |T^+| - |S^+|.
\end{eqnarray*}
But from Theorem \ref{thmMirrorToMedian}, $d(S, T) = |T-S| =  |T^+ - S^+| +  |T^- - S^-|$ so
 \begin{eqnarray}
   S \sq_{P^-} T & \iff &  |T^+ - S^+| +  |T^- - S^-| = |T^+| - |S^+|.
  \label{eqn2}
 \end{eqnarray}

Assume $S^+ \subseteq T^+$.  Since  both $S$ and $T$ are complete closed subsets of $\mathcal{P}$, $T^- \subseteq S^-$.
It follows that $ |T^+ - S^+| = |T^+| - |S^+|$ while $|T^- - S^-|  = 0$.  The right-hand side of  (\ref{eqn2}) is satisfied,  so when $S^+ \subseteq T^+$, it is the case that $S \sq_{P^-} T$.

Conversely, assume $ S \sq_{P^-} T$. From the right-hand side of  (\ref{eqn2}),  $ |T^- - S^-| = (|T^+| - |S^+|) -  |T^+ - S^+|$.  But for any complete closed subsets $S$ and $T$, $ |T^+ - S^+| \ge  |T^+| - |S^+|$ so $ |T^- - S^-| \leq 0$.  Since $ |T^- - S^-|$ cannot be less than $0$, it follows that $ |T^- - S^-| = 0$.  That is, $T^- \subseteq S^-$ so $S^+ \subseteq T^+$.
\end{proof}


For any complete closed subset $T$ of $\mc{P}= (P^- \cup P^+, \leq)$, define the {\it interval} $[P^-, T]$ in $\mc{L}(\mc{P})$ as
$[P^-, T] = \{S:  P^- \sq S \sq T\}$.   Since $P^-$ is the minimum element of $\mc{L}(\mc{P})$, $[P^-,T]$ consists  of $T$ and all the predecessors of $T$ in $\mc{L}(\mc{P})$. From the definition of median semilattices,  the subposet induced by the elements of $[P^-,T]$ is a distributive lattice.  Now, according to Birkhoff's result in Theorem \ref{thmBirkhoff}, there must be a poset that encodes the elements of $[P^-,T]$.   We identify this poset next.

\begin{theorem}
Let $\mathcal{P} = (P^- \cup P^+, \leq)$ and let $T$ be a complete closed subset of $\mathcal{P}$.  Denote as  $\mc{P}_{T^+}$  the subposet of $\mc{P}$ induced by the elements of $T^+$.  There is a one-to-one correspondence between the elements of $[P^-, T]$ in the median semilattice $\mc{L(P)}$ and the closed subsets of $\mc{P}_{T^+}$.   In particular,  $S \in [P^-, T]$ if and only if $S^+$ is a closed subset of $\mc{P}_{T^+}$.
\label{thmChar}
\end{theorem}

Before we prove the theorem, it is helpful to go through an example.  Consider the top right median semilattice  in Figure \ref{figsemis} rooted at $P^-$.  Let  $T = \{\rho_1^+, \rho_2^-, \rho_3^-, \rho_4^+, \rho_5^+\}$ so $T^+ = \{\rho_1^+, \rho_4^+, \rho_5^+\}$.  Thus,  $\mc{P}_{T^+}$ is the subposet of $\mc{P}$ induced by $\rho_1^+, \rho_4^+$ and $ \rho_5^+$.  Its closed subsets are $\emptyset$, $\{\rho_4^+\}$, $\{\rho_5^+\}$, $\{\rho_4^+, \rho_5^+\}$, $\{\rho_1^+, \rho_4^+, \rho_5^+\}$.   On the other hand, the elements of the interval $[P^-,T]$ in the median semilattice correspond to the ``completions'' of the closed subsets of $\mc{P}_{T^+}$; i.e., negative elements are added to the closed subsets so they become complete closed subsets of $\mc{P}$.  They are 
$\{\rho_1^-, \rho_2^-, \rho_3^-, \rho_4^-, \rho_5^-\}$, $\{\rho_1^-, \rho_2^-, \rho_3^-, \rho_4^+, \rho_5^-\}$, $\{\rho_1^-, \rho_2^-, \rho_3^-, \rho_4^-, \rho_5^+\}$, $\{\rho_1^-, \rho_2^-, \rho_3^-, \rho_4^+, \rho_5^+\}$, $\{\rho_1^+, \rho_2^-, \rho_3^-, \rho_4^+, \rho_5^+\}$.


\begin{proof}   Assume $S \in [P^-, T]$.  Then $S \sq T$ in $\mc{L(P)}$.  By Lemma \ref{lemmasubset}, $S^+ \subseteq T^+$. We just have to argue that $S^+$ a closed subset of $\mathcal{P}_{T^+}$.  Let $\rho^+ \in S^+$ and suppose $\sigma^+ \leq \rho^+$ in
 $\mathcal{P}_{T^+}$. Then $\sigma^+ \leq \rho^+$ in $\mc{P}$ too.   Now, $S$ itself is a closed subset of $\mc{P}$ so every predecessor of $\rho^+$, whether positive or negative, is in $S$.  Thus, $\sigma^+ \in S^+$ and, therefore, $S^+$ is a closed subset of $\mc{P}_{T^+}$.

Conversely, let $S^+$ be a closed subset of $\mathcal{P}_{T^+}$.  It follows that $S = S^+ \cup \{\tau^- \in P^-: \tau^+ \not \in S^+\}$ is a complete subset of $\mc{P}$.  We will  show that $S$ is also a closed subset of $\mc{P}$.

\noindent {\bf Case 1:} Let $\rho^+ \in S$.  It means that $\rho^+ \in S^+$.  Since $S^+$ be a closed subset of $\mathcal{P}_{T^+}$, $S^+ \subseteq T^+$.  Furthermore,  $T^+ \subseteq T$. It follows that $\rho^+ \in T$.  Now $T$ is a complete closed subset of $\mc{P}$ so all the predecessors of $\rho^+$, whether positive or negative, are also in $T$.  In particular, (i) if $\sigma^+ \leq \rho^+$, then $\sigma^+ \in T^+$, and (ii) if $\sigma^- \leq \rho^+$, then $\sigma^- \in T^-$.

Consider  $\sigma^+$ so $\sigma^+ \leq \rho^+$.  From (i) and the fact that  $S^+$ is a closed subset of $\mathcal{P}_{T^+}$ implies that if  $\rho^+ \in S^+$ then $\sigma^+ \in S^+$. Therefore $\sigma^+ \in S$.  
Next, consider $\sigma^-$ so $\sigma^- \leq \rho^+$. From (ii),  $\sigma^+ \not \in T^+$ because $T$ is a complete closed subset of $\mc{P}$,  so $\sigma^+ \not \in S^+$.  This means $\sigma^- \in \{\tau^- \in P^-: \tau^+ \not \in S^+\}$ and therefore $\sigma^- \in S$.  
Thus, all the predecessors of $\rho^+$ in $\mc{P}$ are also in $S$.

\noindent {\bf Case 2:} Next, let $\rho^- \in S$ so  $\rho^+ \not \in S^+$.  As a negative rotation, $\rho^-$ can only have negative predecessors  in $\mc{P}$ because of our convention.  Let $\sigma^- \leq \rho^-$.  Then $\rho^+ \leq \sigma^+$ because $\mc{P}$ is a mirror poset.  If $\sigma^+ \in S^+$, we argued above that all its positive predecessors must be in $S^+$, including $\rho^+$.  But $\rho^+ \not \in S^+$ so $\sigma^+ \not \in S^+$.  It follows that $\sigma^- \in  \{\tau^- \in P^-: \tau^+ \not \in S^+\}$ and, hence, $\sigma^- \in S$.

We have now shown that $S$ is a closed subset of $\mc{P}$.  Since $S^+ \subseteq T^+$ because $S^+$ is a closed subset of $\mathcal{P}_{T^+}$, Lemma \ref{lemmasubset} implies that $S \sq T$ in $\mc{L(P)}$ and therefore $S \in [P^-, T]$.
\end{proof}

\noindent {\bf Local Optima.}  Let us now shift the setting back to the SR instance $I$ and its set of stable matchings $M(I)$.  For any subset $M' \subseteq M(I)$, define a stable matching $\lambda$ as a {\it local optimum of $M'$} with respect to the cost function $C$ if $\lambda$ has the least cost among all the stable matchings in $M'$.
 In the next theorem, we will only consider subsets $M'$ of $M(I)$ so that stable matchings in $M'$ form a distributive lattice in some median semilattice.  Finding a local optimum will then reduce to finding a maximum-cost closed subset of some poset.  As we already showed in the SM case, this can be done efficiently.


\begin{theorem}
Suppose $I$ is an SR instance with $2n$ agents, and $\eta$ and $\mu$ are two of its stable matchings.  Let $\mc{L}(I, \eta)$ be the median semilattice obtained by rooting $G(I)$ at $\eta$.  Let  $[\eta, \mu]$ be the interval that contains $\mu$ and all its predecessors in $\mc{L}(I, \eta)$.
Given $I$, the cost function $C$,  $S_\eta$ and $S_\mu$, finding a local optimum of  $[\eta, \mu]$  takes $O(n^4 \log n)$ time.
\label{thmLocalOpt}
\end{theorem}

\begin{proof}  Recall that $\mc{R}'(I)$, the reduced rotation poset of $I$, is a mirror poset.  For ease of notation, let $\mc{R}'= \mc{R}'(I)$.  Orient $\mc{R}'$ at $S_\eta$ by partitioning its element set into $R^- \cup R^+$ where $R^- = S_\eta$. We will now can refer to the non-singular rotations of $I$ as {\it negative} or {\it positive} based on their membership in $R^-$ and $R^+$ respectively.

Given $\mc{R}' = (R^- \cup R^+, \leq)$, root the median graph $G(\mc{R}')$ at $S_\eta$ to create the
 the median semilattice $\mc{L}(\mc{R}')$.   Since $G(\mc{R}')$ is isomorphic to $G(I)$, it is also the case that $\mc{L}(\mc{R}')$ is isomorphic to $\mc{L}(I, \eta)$ since they are obtained by rooting  $G(\mc{R}')$ at $S_\eta$ and $G(I)$ at $\eta$ respectively.
 To prove the theorem, we shall move back and forth between $\mc{L}(\mc{R}')$ and $\mc{L}(I, \eta)$ so we can apply the results we proved earlier.  That is, for a stable matching $\lambda$, we will sometimes think of it as $S_\lambda$, the complete closed subset of $\mc{R}'$ that corresponds to $\lambda$.



Consider $\lambda \in [\eta, \mu]$.  From Theorem \ref{thmSREgal}, its cost is $$C(\lambda) \, = \, C(\eta) - \!\!\!\sum_{\rho \in S_\lambda-S_\eta} \!\!\!\!C(\rho) \, = \, C(\eta) - \!\!\!\sum_{\rho \in S^+_\lambda} \!C(\rho) \, = \, C(\eta) - C(S^+_\lambda)$$
where the second equality follows from the fact that $S_\eta$ consists of all the negative rotations.   Thus, $\lambda^*$ is a local optimum of $[\eta, \mu]$ if and only if  $C(S^+_\lambda) \leq C(S^+_{\lambda^*})$ for any $\lambda \in [\eta, \mu]$.

Now, $\lambda \in [\eta, \mu]$ in $\mc{L}(I, \eta)$ if and only if $S_\lambda \in [S_\eta, S_\mu]$ in $\mc{L}(\mc{R}')$. According to Theorem \ref{thmChar},  $S_\lambda \in [S_\eta, S_\mu]$ in $\mc{L}(\mc{R}')$ if and only if $S^+_\lambda$ is a closed subset of $\mc{R}'_{S^+_\mu}$. Thus, to find a local optimum of $[\eta, \mu]$, we simply have to find a maximum-cost closed subset of $\mc{R}'_{S^+_\mu}$.

Again, according to Theorem \ref{thmChar}, this maximum-cost closed subset of $\mc{R}'_{S^+_\mu}$ contains the positive elements of some complete closed subset $S^* \in  [S_\eta, S_\mu]$.  In turn, $S^*$ corresponds to some stable matching $\lambda^*$ in $[\eta, \mu]$, which has to be a local optimum of $[\eta, \mu]$ since $C(\lambda^*) \leq C(\lambda)$ for any $\lambda \in  [\eta, \mu]$.

 Let us now analyze the running time of the algorithm we have just outlined for finding a local optimum in the interval $[\eta, \mu]$.  Given SR instance $I$, the cost function $C$, $S_\eta$ and $S_\mu$,  we have to first compute a directed acyclic graph $H(I)$ whose transitive closure is $\mc{R}'(I)$.  From Corollary \ref{corFindMatching}, this can be done in $O(n^3 \log n)$ time.  Furthermore,  $H(I)$ has $O(n^2)$ vertices and $O(n^2)$ edges.  Next, make the rotations in $S_\eta$ be the negative elements and the rotations not in $S_\eta$ be the positive elements.  Extract the subgraph of $H(I)$ induced by the rotations in ${S^+_\mu}$.  The transitive closure of this subgraph is $\mc{R}_{S^+_\mu}$ because negative rotations cannot be successors of positive rotations.   Finally, compute a maximum-cost closed subset of $\mc{R}_{S^+_\mu}$ using the algorithm Irving et al.~used for Corollary \ref{corEgal}.   The running time analysis is based on the fact that when the representation of the rotation poset has $|V|$ vertices and $|E|$ edges, the minimum cut can be obtained in $O(|V| |E| \log n)$ time.  In our case $|V| = O(n^2)$ and $|E| = O(n^2)$ so a maximum-cost closed subset  of $\mc{R}_{S^+_\mu}$ can be obtained in $O(n^4 \log n)$ time.

Once the maximum-cost closed subset of $\mc{R}_{S^+_\mu}$ is found, we add the negative elements of the dual pairs missing from the subset to create $S^*$.  According to Corollary \ref{corFindMatching}, given $S^*$, the stable matching corresponding to $S^*$ can  be obtained in $O(n^2)$ time.
\end{proof}

%% file: MaximalElements.tex
\section{Maximal Elements in a Median Semilattice}
\label{sec:MaximalElements}

In the previous section, 
we showed how the technique for finding an optimal stable matching in SM instances can be extended to finding a {\it local} optimal stable matching in SR instances under certain conditions.  In particular, in the median semilattice $\mc{L}(I, \eta)$,
we can efficiently find  a local optimum of the interval $[\eta, \mu]$ for any $\mu$ given $S_\eta$ and $S_\mu$.    We take advantage of this fact as follows:
Consider the {\it maximal elements} of  $\mc{L}(I, \eta)$.  Assume they are $\mu_1, \mu_2, \hdots, \mu_r$.   Notice that $\cup_{i=1}^r [\eta, \mu_i]$ contains all the stable matchings of $I$.  Furthermore, an optimal stable matching of $I$ must lie in one of the intervals.  Suppose  $\mu_i^*$ is the local optimum of $ [\eta, \mu_i]$ for $i = 1, \hdots, r$.  Then  the stable matching with the least cost among $\{\mu^*_i, i = 1, \hdots, r\}$ has to be an optimal stable matching of $I$.  To implement this plan we must identify the maximal elements of $\mc{L}(I, \eta)$, their corresponding closed subsets and bound their number.


We will once more work with mirror posets and median semilattices to establish our results.  Suppose $\mc{P}$ has been oriented as  $(P^- \cup P^+, \leq)$
and  $HD(\mathcal{P})$ is its Hasse diagram.  
Recall that an edge of $HD(\mathcal{P})$ is a {\it crossing edge}  if one endpoint of the edge is in $P^-$ and another is in $P^+$.  Crossing edges come in pairs because $\mc{P}$ is a mirror poset. As we shall see, they  play an important role.

\begin{proposition}
	Let $\sigma^-, \rho^+$ be elements of the mirror poset $\mathcal{P}= (P^- \cup P^+, \leq)$ such that $\sigma^- < \rho^+$.  There must be a crossing edge $(\alpha^-, \beta^+)$  so that  $\sigma^- \leq \alpha^- < \beta^+ \leq \rho^+$.
	\label{propcrossing}

\end{proposition}

\begin{proof}
	Since $\sigma^- < \rho^+$, there is a directed path $\tau_1, \tau_2, \hdots, \tau_m$ with $\tau_1 = \sigma^-$ and $\tau_m = \rho^+$ in $HD(\mathcal{P})$.   Since $\tau_1  \in P^-$ and $\tau_m \in P^+$,  the directed path must cross from $P^-$ to $P^+$; i.e.,  there is some $\tau_i$ so that $\tau_i \in P^-$ and $\tau_{i+1} \in P^+$.   Thus, $(\tau_i, \tau_{i+1})$ is a crossing edge of $\mathcal{P}$.   But, additionally, the path from $\tau_1$ to $\tau_m$ forms a chain in $\mathcal{P}$ so  $\tau_1 \leq \tau_i < \tau_{i+1} \leq \tau_m$.
\end{proof}

We say that the pair {\it $\{\alpha, \beta\}$ is involved in a crossing edge of $\mathcal{P}= (P^- \cup P^+, \leq)$} if $(\alpha^-, \beta^+)$ and $(\beta^-, \alpha^+)$ are crossing edges of the oriented mirror poset.   For $A \subseteq P^- \cup P^+$, let $Pred(A)$ contain all the predecessors of the elements of $A$ in $\mc{P}$.

\begin{lemma}
	Assume that $\mathcal{P} = (P^- \cup P^+, \leq)$ has $2k$ crossing edges involving the pairs  $\{\alpha_{i}, \beta_{i}\}$, $i = 1, \hdots, k$.  Every complete closed subset $S$ of $\mathcal{P}$ can be expressed as $S = A_S \cup Pred(A_S) \cup B_S$ where

	\noindent (i) $A_S =  \bigcup_{i=1}^k A_{S,i}$ and $A_{S,i} \in \{ \{\alpha_{i}^+, \beta_{i}^- \}, \{\alpha_{i}^-, \beta_{i}^+ \}, \{\alpha_{i}^-, \beta_{i}^- \} \}$ for each $i$ and

	\noindent (ii) $B_S$ has one element from each dual pair $\{\rho^-, \rho^+\}$ such that neither $\rho^-$ nor $\rho^+$ are in $A_S \cup Pred(A_S)$.
	\label{lemmarep}
\end{lemma}

\begin{proof}
	By assumption, for each pair $\{\alpha_i, \beta_i\}$,  $\alpha_{i}^- < \beta_{i}^+$ and  $\beta_{i}^- < \alpha_{i}^+$ in $\mc{P}$.  Since $S$ is a complete closed subset, if $\beta_{i}^+ \in S$, then so is $\alpha_{i}^- \in S$.  Similarly, if  $\alpha_{i}^+ \in S$, then so is $\beta_{i}^- \in S$.  It follows that for each $i$,  one of  $\{\alpha_{i}^+, \beta_i^- \}, \{\alpha_{i}^-, \beta_{i}^+ \}$ or $\{\alpha_{i}^-, \beta_{i}^- \}$ is a subset of $S$.  We shall call this particular subset $A_{S,i}$ and let $A_S =  \bigcup_{i=1}^k A_{S,i}$.


	Now, $A_S \subseteq S$ and $S$ is a closed subset of $\mathcal{P}$ so $ Pred(A_S) \subseteq S$ too.    Additionally, because $S$ is a complete closed subset,  $S - (A_S \cup Pred(A_S))$ must contain one element from each dual pair $\{\rho^-, \rho^+\}$ such that neither $\rho^-$ nor $\rho^+$ are in $A_S \cup Pred(A_S)$.  Let $B_S = S - (A_S \cup Pred(A_S))$.  Thus, $S = A_S \cup Pred(A_S) \cup B_S$.
\end{proof}



We now describe the maximal elements of $\mc{L(P)}$.

\begin{theorem}
	Assume that $\mathcal{P} = (P^- \cup P^+, \leq)$ has $2k$ crossing edges involving the pairs  $\{\alpha_{i}, \beta_{i}\}$, $i = 1, \hdots, k$.  Every maximal element $T$ of $\mc{L}(\mathcal{P})$ can be expressed as $T = A_T \cup Pred(A_T) \cup C_T$ where

	\noindent (i) $A_T =  \bigcup_{i=1}^k A_{T,i}$ and $A_{T,i} \in \{ \{\alpha_{i}^+, \beta_{i}^- \}, \{\alpha_{i}^-, \beta_{i}^+ \}, \{\alpha_{i}^-, \beta_{i}^- \} \}$ for each $i$ and

	\noindent (ii) $C_T$ contains $\rho^+$ of each dual pair $\{\rho^-, \rho^+\}$ such that neither $\rho^-$ nor $\rho^+$ are in $A_T \cup Pred(A_T)$.

	\label{thmmaximalrep}
\end{theorem}

\begin{proof}
	The elements of $\mc{L(P)}$, including its maximal elements, are complete closed subset of $\mathcal{P}$.  Thus, from Lemma \ref{lemmarep},  a maximal element $T = A_T \cup Pred(A_T) \cup B_T$ where
	$B_T$ has one element from each dual pair $\{\rho^-, \rho^+\}$ such that neither $\rho^+$ nor $\rho^-$ are in $A_T \cup Pred(A_T)$.

	Now, define set  $U$ as $U= A_T \cup Pred(A_T) \cup C_T$.    We shall argue next that $U$ is also a complete closed subset of $\mathcal{P}$.  Since $T$ is a complete subset of $\mathcal{P}$, the set $U$ is also one because both $B_T$ and $C_T$ contain exactly one element from each dual pair missing a representative in $A_T \cup Pred(A_T)$.   We just have that to show that $U$ is a closed subset as well.  Suppose not.  Then some $\rho^+ \in C_T$ is missing a predecessor in $U$, say $\sigma^-$ or $\sigma^+$.

	\noindent {\bf Case 1: $\sigma^-$ is missing from $U$.}  Since $\sigma^- < \rho^+$,  by Proposition \ref{propcrossing},  there is pair $\{\alpha_i, \beta_i\}$ involved in a crossing edge of $H(\mathcal{P})$ so that (i) $ \sigma^- \leq \alpha_i^- < \beta_i^+ \leq \rho^+$ or  (ii) $\sigma^- \leq \beta_i^- < \alpha_i^+ \leq \rho^+.$   By construction, $A_T$ contains $\alpha_i^-$ or $\beta_i^-$. We have four subcases to consider.  For each one, we will argue that some contradiction arises so case 1 cannot be true.

	If (i) is true and $\alpha_i^- \in A_T$ or (ii) is true and $\beta_i^- \in A_T$, then $\sigma^- \in A_T \cup Pred(A_T)$ and cannot be missing from $U$.   On the other hand,  if (i) is true and $\beta_i^- \in A_T$,  we know that $\rho^- \leq  \beta_i^- < \alpha_i^+ \leq \sigma^+$ because $\mathcal{P}$ is a mirror poset.  This implies that $\rho^- \in A_T \cup Pred(A_T)$, contradicting the fact that $\rho^+ \in C_T$.   Similarly,  if (ii) is true and  $\alpha_i^- \in A_T$, we know that $\rho^- \leq  \alpha_i^- < \beta_i^+ \leq \sigma^+$ so $\rho^- \in A_T \cup Pred(A_T)$, again contradicting the fact that $\rho^+ \in C_T$.

	\noindent {\bf Case 2: $\sigma^+$ is missing from $U$.}   Since $\sigma^+ < \rho^+$, we have $\rho^- < \sigma^-$.  Now, $\sigma^+ \not \in U$ implies $\sigma^+ \not \in C_T$.  This means that $\sigma^- \in A_T \cup Pred(A_T)$ so $\rho^- \in A_T \cup Pred(A_T)$;  otherwise, an element of $A_T$ is missing a predecessor.   But this contradicts the fact that $\rho^+ \in C_T$.  

	Our argument shows that $U = A_T \cup Pred(A_T) \cup C_T$ is a  complete closed subset of $\mathcal{P}$.  Now, $T^+$ consists of the positive elements in $A_T \cup Pred(A_T)$ and in $B_T$.  But the positive elements of $B_T$ are all in $C_T$.  Thus, $T^+ \subseteq U^+$ and therefore $T \sq U$ in $\mc{L(P)}$ by Lemma \ref{lemmasubset}.  But $T$ is a maximal element of $L(\mathcal{P})$. So it must be the case that $B_T = C_T$ and $T = A_T \cup Pred(A_T) \cup C_T$.
\end{proof}

Informally,  Theorem \ref{thmmaximalrep} implies that every maximal element $T$ of $\mc{L}(\mathcal{P})$ is completely determined by which of the crossing edges of $\mathcal{P}$ it uses.  For each pair $\{\alpha_{i}, \beta_{i}\}$, $i = 1, \hdots, k$, it has three options: use the edge  $(\alpha^+_{i}, \beta^-_{i})$ or the edge $(\alpha_{i}^-, \beta_{i}^+)$ or do not use any of them.  The options correspond to setting $A_{T,i}$ to $\{\alpha^+_{i}, \beta^-_{i}\}$, $\{\alpha^-_{i}, \beta^+_{i}\}$ and $\{\alpha^-_{i}, \beta^-_{i}\}$ respectively.  Once all the choices have been made, the rest of the elements in $T$ are set:  they consist of the missing predecessors of the elements in $A_T$ and the positive elements of dual pairs not represented in $A_T$.
As an example, consider the maximal element $T = \{\rho_1^+, \rho_2^-, \rho_3^-, \rho_4^+, \rho_5^+\}$ of the median semilattice in Figure \ref{figsemis}. It uses the crossing edges $(\rho_2^-, \rho_4^+)$, $(\rho_3^-, \rho_4^+)$, and $(\rho_3^-, \rho_5^+)$ of the corresponding mirror poset so $A_T = \{\rho_2^-, \rho_3^-, \rho_4^+, \rho_5^+\}$.  Now $A_T$ is already a closed subset and has no missing predecessors, but it has no  representatives from $\{\rho^-_1, \rho^+_1\}$ so $\rho^+_1$ is added to $A_T$ to obtain $T$.  We can now bound the number of maximal elements in $\mc{L(P)}$.

\begin{corollary}\label{cor:crossing-maximal}
	Assume that $\mathcal{P} = (P^- \cup P^+, \leq)$ has $2k$ crossing edges.  Then $\mc{L(P)}$ has at most $3^k$ maximal elements.
\end{corollary}

\begin{proof}
	From Theorem \ref{thmmaximalrep}, every maximal element  $T $ of $L(\mathcal{P})$ can be written as $T= A_T \cup Pred(A_T) \cup C_T$ where $A_T =  \bigcup_{i=1}^k A_{T,i}$.  Each $A_{T,i}$ can be one of three possible sets.  We already noted that once each $A_{T,i}$ is specified, the rest of the elements of $T$ are completely determined. 
    It follows that there can be at most $3^k$ distinct forms of $T$.
\end{proof}

We note that the converse of Theorem \ref{thmmaximalrep} does not hold.  That is,  by specifying the choices for each $A_{T,i}$, combining it with $Pred(A_T)$ and the positive elements of the dual pairs with missing representatives, the resulting set $T$ may not be a maximal element of $L(\mathcal{P})$.  One reason is because when  the sets for the $A_{T,i}$, $i = 1, \hdots, k$, are not chosen carefully, the positive and negative elements of a dual pair can appear in $A_T$ so $T$ is not a complete subset of $L(\mathcal{P})$.  But even if $T$ is a complete closed subset of $L(\mathcal{P})$, it may not be a maximal element of $L(\mathcal{P})$.  Thus, the upper bound in Corollary \ref{cor:crossing-maximal} is not necessarily tight.

\begin{corollary}\label{cor:optimal-matching}
	Let $I$ be an SR instance with $2n$ agents. Let $\mc{L}(I, \eta)$ be the median semilattice obtained by rooting $G(I)$ at the stable matching $\eta$. Suppose that when $\mc{R}'(I)$ is oriented at $S_\eta$, the orientation has at most $2k$ crossing edges.  Then $\mc{L}(I, \eta)$ has at most $3^k$ maximal elements.  Furthermore, given $I$, cost function $C$ and $S_\eta$,  computing an optimal stable matching of $I$ can be done in $3^k n^{O(1)}$ time.
\end{corollary}

\begin{proof}  
	The fact that $\mc{L}(I, \eta)$ has at most $3^k$ maximal elements follows from Corollary \ref{cor:crossing-maximal} immediately.  So we consider the problem of computing an optimal stable matching of $I$.

	Construct the directed acyclic graph $H(I)$ whose transitive closure is $\mc{R}'(I)$ as described in Theorem \ref{thm:sr-poset}.   Do a transitive reduction of $H(I)$ so that only the edges in the Hasse diagram of $\mc{R}'(I)$ are left.  Orient $\mc{R}'(I)$ at $R^- = S_\eta$ and then identify all the pairs  involved in a crossing edge of $(R^- \cup R^+, \leq)$.  These steps can be done in $n^{O(1)}$ time.


	Assume $\{\alpha_i, \beta_i\}$, $i = 1, \hdots,  k$,  are involved in a crossing edge of the oriented $\mc{R}'(I)$.   Consider the $3^k$ possible forms of a maximal element $T$ as described in Theorem \ref{thmmaximalrep}.  Discard those that are not complete closed subsets of $\mc{R}'(I)$. This step can be completed in $3^k n^{O(1)}$ time since $\mc{R}'(I)$ has only $O(n^2)$ elements.    For each valid $T$,  compute a locally optimal stable matching in $[\eta, \mu_T]$ where $\mu_T$ is the stable matching corresponding to $T$.
	According to Theorem \ref{thmLocalOpt}, this can be done in $O(n^4 \log n)$ time.    Finally, among the locally optimal stable matchings, output the one that has the least cost. It must be an optimal stable matching of $I$. This step also takes $3^k n^{O(1)}$ time.
\end{proof}

%% file: min-crossings.tex
\section{Minimum Crossing Edges}\label{sec:min-crossing}

In the previous sections, we showed that the structure of the median semilattice $\mc{L}(I, \eta)$ can be understood in terms of the reduced rotation poset $\mc{R}'(I)$ when the latter is oriented at $S_
	\eta$.  In particular, the {\it number} of crossing edges of the oriented $\mc{R}'(I)$ plays a very important role---it can be used to bound the number of maximal elements in $\mc{L}(I, \eta)$ and, consequently, bound the running time for finding an optimal stable matching of $I$.

But $\mc{L}(I, \eta)$ is one of many possible semilattices that can arise from the median graph $G(I)$. By choosing a different root for $G(I)$ and, in parallel, a different base for $\mc{R}'(I)$,  the resulting orientation of $\mc{R}'(I)$ can have far fewer crossing edges.  We can then substantially improve the running time for computing an optimal stable matching of $I$.  Thus, we consider the following optimization problem:

\noindent \textsc{Minimum Crossing Orientation (MCO):} {\it Given a mirror poset $\mc{P}$, find an orientation of $\mc{P}$ so that it has the least number of crossing edges among all the orientations of $\mc{P}$.}

Let $\mco(\mc{P})$ denote the number of crossing edges in an {\it optimal} orientation of $\mc{P}$, which we also refer to as   the {\it minimum crossing distance of $\mc{P}$}. Additionally,  when $\mc{P}$ is the reduced rotation poset of SR instance $I$, we  call $\mco(\mc{P})$ the {\it minimum crossing distance of $I$}.  We begin by characterizing the mirror posets with minimum crossing distance $0$.

\begin{lemma}
	Let $\mc{P}$ be a mirror poset.  Then $\mco(\mc{P}) = 0$ if and only if $\mc{P}$ has an orientation $(P^- \cup P^+, \leq)$ so that $\mc{L}(\mc{P})$ is a distributive lattice.  Furthermore, given a directed acyclic graph $H$ whose transitive closure contains exactly the relations in $\mc{P}$, we can determine if $\mco(\mc{P}) = 0$  in time linear in the size of $H$.
	\label{prop:MCO}
\end{lemma}


To prove Lemma \ref{prop:MCO}, we shall make use of the next two results.  Recall that in a mirror poset $\mc{P} = (P, \leq)$, a subset $U \subseteq P$ is {\it partially complete} if it contains {\it at most} one element of every dual pair of elements in $\mc{P}$.

\begin{lemma}
	[Lemma 3 in \cite{Cheng2011Stable}] Let $\mathcal{P} = (P, \leq)$ be a mirror poset, and let $U$ be one of its partially complete closed subset.  Then there is a complete closed subset $S$ so that $U \subseteq S$.  
	\label{lemmapartial}
\end{lemma}

\begin{corollary}
	Let $\mathcal{P}= (P, \leq)$ be a mirror poset.  For every element $\rho$ of $\mathcal{P}$,  there is a complete closed subset $S_\rho$ that contains $\rho$ and its predecessors.
	\label{cor:pred}
\end{corollary}

\begin{proof}
	Let $\Pred(\rho)$ contain all the predecessors of $\rho$.  Consider $U = \{\rho\} \cup \Pred(\rho)$.   We now prove that $U$ is a partially complete closed subset.  Let $\tau \in U$ and $\sigma < \tau$.  Then  $\sigma <  \tau \leq \rho$.  It follows that $\sigma \in \Pred(\rho)$.  That is, $U$ is a closed subset.   Suppose $U$ is not partially complete.  Then there is a pair of dual elements $\sigma$ and $\overline{\sigma}$ in $U$.  Thus, $\sigma \leq \rho$ and $\overline{\sigma} \leq \rho$.  But $\mathcal{P}$ is a mirror poset so $\overline{\rho} \leq \overline{\sigma}$ and $\overline{\rho} \leq \sigma$.  This implies that $\overline{\rho} \leq \rho$, violating the property that an element and its dual are incomparable in a mirror poset.  Hence, $U$ is partially complete.  From Lemma \ref{lemmapartial}, there must be a complete closed subset $S_\rho$ so that $U \subseteq S_\rho$.
\end{proof}


\begin{proof}[Proof of Lemma \ref{prop:MCO}.] Assume $\mco(\mc{P}) = 0$ and $(P^- \cup P^+, \leq)$ is an orientation of $\mc{P}$ that has no crossing edges.  By Corollary \ref{cor:crossing-maximal}, $\mc{L}(\mc{P})$ has at most $3^0 = 1$ maximal element; i.e., $\mc{L}(\mc{P})$ has a maximum element.  Thus, $\mc{L}(\mc{P})$ is a distributive lattice by Proposition \ref{propmax}.


	On the other hand, assume $(P^- \cup P^+, \leq)$ is an orientation of $\mc{P}$ and  $\mc{L}(\mc{P})$ is a distributive lattice. According to Proposition \ref{propmax},  $\mc{L}(\mc{P})$ has a maximum element $T$. We will now argue that $T = P^+$.  Suppose not and $\rho^- \in T$.  Let $U$ contain $\rho^+$ and its predecessors.  By Corollary \ref{cor:pred}, $\mc{P}$ has a complete closed subset $S_{\rho^+}$ so that $U \subseteq S_{\rho^+}$.  Since $T$ is the maximum element of $\mc{L}(\mc{P})$, $S^+_{\rho^+} \subseteq T^+$ by
	Lemma \ref{lemmasubset}.  It follows that $\rho^+ \in T$.  But $\rho^- \in T$ too, contradicting the fact that $T$ is  a complete closed subset of $\mc{P}$.  Hence, $T = P^+$. In other words, $P^+$ is also a complete closed subset of $\mc{P}$ so $(P^- \cup P^+, \leq)$ has no crossing edges.

	Next, we describe an algorithm for determining if $\mco(\mc{P}) = 0$.  Consider $H$ whose transitive closure contains exactly the relations in $\mc{P}$, and let $\hat{H}$ be its undirected version.  Run a graph traversal (BFS or DFS) on $\hat{H}$. For each vertex of $\hat{H}$, record the connected component that contains the vertex.  If some dual pair  $\rho$ and $\ol{\rho}$  lie in the same connected component of $\hat{H}$, return ``no"; otherwise, return ``yes."  Clearly, the algorithm runs in time linear in the size of $\hat{H}$ and therefore of $H$.  We now argue its correctness.

	Again, assume $\mco(\mc{P}) = 0$ and $(P^- \cup P^+, \leq)$ is an orientation of $\mc{P}$ that has no crossing edges. For every dual pair $\{\rho, \ol{\rho}\}$, one element lies in $P^-$ and another in $P^+$.  Since there are no edges between $P^-$ and $P^+$ in the Hasse diagram of $\calP$, there is no path  between $\rho$ and $\ol{\rho}$ in $\hat{H}$ so the two elements have to lie in different connected components of $\hat{H}$. It follows that the algorithm will return ``yes."

	Conversely, assume the algorithm returned ``yes."   Pick an element $\rho$.  Notice that in $\hat{H}$,  there is a path  from $\rho$ to another element $\sigma$ if and only if there is a path from $\ol{\sigma}$ to $\ol{\rho}$ because $\mc{P}$ is a mirror poset.  Assume $\rho$ lies in connected component $C_1$ and $\ol{\rho}$ lies in connected component $C_2$.  Since the algorithm returned ``yes,'' $C_1 \neq C_2$.  Furthermore,
	the vertices in $C_1$ are duals of the vertices in $C_2$ according to our discussion.
	Put the vertices of $C_1$ into $P^-$ and the vertices of $C_2$ into $P^+$.  Do this again until all the elements of $\mc{P}$ have been assigned to $P^-$ or $P^+$.   Then the duals of the elements in $P^-$ are in $P^+$ and vice versa.  There are also no edges between $P^-$ and $P^+$ in the Hasse diagram of $\calP$ since their respective elements lie in different connected components of $\hat{H}$.  Thus, $P^-$ and $P^+$ are both complete closed subsets of $\mc{P}$ so $(P^- \cup P^+, \leq)$ is an orientation of $\mc{P}$ with no crossing edges.
\end{proof}

Recall that $G(\mc{P})$ is the median graph formed by the complete closed subsets of $\mc{P}$.  The parameter $\mco(\calP)$ gives a quantitative measure of how structurally similar $G(\mc{P})$ is to the covering graph of a distributive lattice.  According to Proposition \ref{prop:MCO},
$\mco(\calP) = 0$ if and only if $G(\mc{P})$ can be rooted at some complete closed subset $S$ and the result is a distributive lattice. The said $S$ is the base that should be used so that the orientation of $\mc{P}$ has no crossing edges.
More generally, Corollary~\ref{cor:crossing-maximal} implies that if $\mco(\calP) = 2k$, then $G(\mc{P})$ is the union of the covering graphs of at most $3^k$ distributive lattices.

In the remainder of this section, we analyze the computational complexity of MCO. We first prove that MCO is NP-hard.  We then show that MCO is fixed-parameter tractable when parameterized by $\mco(\calP)$.  Finally, we prove our main result, Theorem~\ref{thm:main}, which states that finding an optimal stable matching of SR instance $I$ is fixed-parameter tractable when parameterized by the minimum crossing distance of $I$.

When given a mirror poset $\mc{P} = (P, \leq)$, we will use $(P^-, P^+) $ to denote a partition of $P$.  We shall say that the partition is {\it perfect} if both $P^-$ and $P^+$ are complete subsets of $\mc{P}$.  If $P^-$ is additionally a closed subset of $\mc{P}$, then $(P^-, P^+)$ is an orientation of $\mc{P}$. Thus, like the earlier sections of the paper, we will consider the crossing edges of $(P^-, P^+)$.





\subsection{Hardness of MCO}



Given a graph $G = (V, E)$, a \dft{vertex cover} $A \subseteq V$ is a set of vertices such that every edge $e = \set{u, v} \in E$ has a nontrivial intersection with $A$. The {\it minimum vertex cover problem} (MVC) seeks a vertex cover of minimum cardinality.
The decision variant---deciding if $G$ admits a vertex cover of a given size ---is in one of Karp's original list of NP-complete problems~\cite{Karp2010Reducibility}. Subsequently, Garey, Johnson, and Stockmeyer~\cite{Garey1976Some} showed that MVC is NP-hard even when restricted to graphs of maximum degree $3$. A straightforward reduction shows that Garey et al.'s result further implies the following.

\begin{fact}\label{fact:mvc}
	Minimum Vertex Cover (MVC) is NP-hard even when restricted to $3$-regular graphs (i.e., graphs in which every vertex has degree $3$).
\end{fact}

For completeness, we prove Fact~\ref{fact:mvc} in the appendix. We will now use it to show that MCO is NP-hard.

\begin{theorem}\label{thm:mco-hard}
	The minimum crossing orientation problem (MCO) is NP-hard.
\end{theorem}
\begin{proof}
	By Fact~\ref{fact:mvc}, it suffices to give a polynomial time reduction from MVC on $3$-regular graphs to MCO. 
	Given a graph $G = (V, E)$, form the poset $\calP = (P, \leq)$ where $P = \set{v, \ol{v} \sucht v \in V}$, and for each edge $\set{u, v} \in E$, $v \leq \ol{u}$ and $u \leq \ol{v}$.
	It is easy to check that $\calP$ is a mirror poset and its Hasse diagram has edge set $E_\calP = \set{(u, \ol{v}), (v, \ol{u}) \sucht \set{u, v} \in E}$.

	For each subset $U \subseteq V$ of vertices of $G$, associate with it a perfect partition $(P^-_U, P^+_U)$ of $P$ where
	\begin{equation}\label{eqn:subset-orientation}
		\begin{aligned}
			P^-_U ={} & \set{u \sucht u \in U} \cup \set{\ol{v} \sucht v \in V \setminus U}  \text{ and symmetrically} \\
			P^+_U ={} & \set{\ol{u} \sucht u \in U} \cup \set{v \sucht v \in V \setminus U}
		\end{aligned}
	\end{equation}
	It is clear that the correspondence $U \leftrightarrow (P^-_U, P^+_U)$ is a bijection between subsets of $V$ in $G$ and the perfect partitions of $\calP$.

	\begin{description}
		\item[Claim 1.] Let $U \subseteq V$.  Then $U$ is a vertex cover of $G$ if and only if  $(P^-_U, P^+_U)$ is an orientation of $\calP$.

		\item[Proof of Claim 1.] First, suppose $U$ is a vertex cover. That is, for every $\set{u, v} \in E$, $u \in U$ or $v \in U$ (or both). We will show that $(P^-_U, P^+_U)$ is an orientation of $\mc{P}$---i.e., $P^-_U$ is a complete closed subset of $\calP$ so none of the elements in $P^-_U$ have a predecessor in $P^+_U$.   Suppose to the contrary that there is some $(u, \ol{v}) \in E_\calP$ with $u \in P^+_U$ and $\ol{v} \in P^-_U$. By the definition of $\calP$, we must have $\set{u, v} \in E$. By the definition of $P^+_U$, we have $u \notin U$, and by the definition of $P^-_U$, we have $v\notin U$. Thus, the edge $\set{u, v}$ is not covered by $U$ so $U$ is not a vertex cover, a contradiction.

		Conversely, let $(P^-_U, P^+_U)$ be an orientation of $\mc{P}$. Suppose $\set{u, v} \in E$ is an edge in $G$. Since $(P^-_U, P^+_U)$ is an orientation, the associated edge $(u, \ol{v}) \in E_\calP$ does not cross from $P^+_U$ to $P^-_U$. That is,  either (i) $u \not \in P^+_U$ so $\ol{u} \in P^+_U$ and therefore $u \in U$ or  (ii) $\ol{v} \not \in P^-_U$ so $v \in P^-_U$ and therefore $v \in U$. Thus, the edge $\set{u, v}$ is covered by $U$. Since $e$ was chosen arbitrarily, $U$ is a vertex cover of $G$, and Claim~1 follows.
	\end{description}



	In the next claim, we will show that when $G$ is a $3$-regular graph and $U$ is a vertex cover of $G$, the number of crossing edges of the orientation  $(P^-_U, P^+_U)$ of $\mc{P}$ can be computed from the size of $U$.


	\begin{description}
		\item[Claim 2.] Suppose $G = (V, E)$ is a $3$-regular graph with $n$ vertices. Let $U \subseteq V$ be a vertex cover of $G$. Then $(P^-_U, P^+_U)$  has $6|U| - 3n$ crossing edges.

		\item[Proof of Claim 2.]  Let $U$ be a vertex cover of $G$. By Claim 1,  $(P^-_U, P^+_U)$ is an orientation of $\mc{P}$.  For an edge $e$ of $G$,  we say that $U$ {\it double-covers} $e$  if both endpoints of $e$ are in $U$; otherwise, $U$ {\it single-covers} $e$.

		Let $e =\set{u, v}$.  Notice that $U$ double-covers $e$ if and only if $u, v \in P^-_U$ and $\ol{u}, \ol{v} \in P^+_{U}$.  The latter is true if and only if $(u, \ol{v})$ and $(v, \ol{u})$ are crossing edges of $(P^-_U, P^+_U)$.  Thus, $U$ double-covers $k$ edges of $G$ if and only if the orientation $(P^-_U,  P^+_U)$ of $\mc{P}$ has $2k$ crossing edges. (In particular, if no edge is double-covered---i.e., $U$ is an independent set---then $(P^-_U,  P^+_U)$ induces no crossing edges.)

		Let $E_1$ and $E_2$ denote the set containing the edges single-covered and double-covered by $U$ respectively.  Since $U$ is a vertex cover of $G$, its complement $V \setminus U$ is an independent set.  Moreover, every edge $e \in E_1$ is incident to $V \setminus U$.  Now,  each vertex $v \in V \setminus U$ has degree $3$, so
		\begin{equation}\label{eqn:E1}
		    \abs{E_1} = 3 \abs{V \setminus U} = 3n - 3 \abs{U}.
		\end{equation}
		To compute $\abs{E_2}$, observe that $\abs{E_1} + \abs{E_2} = \abs{E} = \frac{3}{2} n$ because $G$ is a $3$-regular graph. Applying~(\ref{eqn:E1}),   $\abs{E_2} = 3 \abs{U} - \frac 3 2 n$. Thus, the  orientation $(P^-_U, P^+_U)$ has $2 \abs{E_2} = 6 \abs{U} - 3 n$ crossing edges. 
	\end{description}
	The reduction from the $3$-regular graph $G$ to $\mc{P}=(P, \leq)$ takes time linear in the size of $G$. From Claim 2,  $U$ is a minimum vertex cover of $G$ if and only if $(P^-_U \cup P^+_U, \leq)$ is a minimum crossing orientation of $\mc{P}$.  Since MVC for 3-regular graphs is NP-hard, it follows that MCO is also NP-hard.
\end{proof}

\subsection{An FPT Algorithm for MCO}

We now show that MCO is fixed-parameter tractable (FPT) with respect to the parameter $k=\mco(P)$
via a reduction from MCO to the {\it Almost
		2-SAT}  problem.

Recall that a $2$-SAT instance consists of $n$  variables $x_1, x_2, \ldots, x_n$ and $m$ clauses $\phi_1, \phi_2, \ldots, \phi_m$ (possibly with duplicate clauses), where each clause is a disjunction of two literals. The Max-$2$-SAT problem seeks  a Boolean assignment of the $n$ variables  that maximizes the number of satisfied clauses.
Max-$2$-SAT is NP-hard~\cite{Garey1976Some}, but in a celebrated result, Razgon and O'Sullivan~\cite{Razgon2009Almost} showed that Max-$2$-SAT admits an FPT algorithm, parameterized by the number of \emph{un}satisfied clauses.  Formally, an input to the \dft{Almost $2$-SAT} problem is a triple $(X, \Phi, k)$ where $X$ consists of $n$ variables, $\Phi$ consists of $m$ two-literal clauses (possibly with duplicate clauses) and $k$ is an integer.  The output is either an assignment of Boolean values to the variables in $X$ such that at most $k$ clauses in $\Phi$ are {\it not} satisfied, or the (correct) assertion that no such assignment exists.

\begin{theorem}[Razgon and O'Sullivan~{\cite{Razgon2009Almost}}]\label{thm:a2sat}
	Suppose $(\Phi, X)$ is a 2-SAT instance and $k$ is the minimum number of unsatisfied clauses in $\Phi$ over all Boolean assignments to $X$. Then a Boolean assignment that does not satisfy k clauses in $\Phi$ can be found in time $O(15^k k m^3)$.
\end{theorem}

We now make use of Theorem~\ref{thm:a2sat}  to prove the next result.

\begin{theorem}\label{thm:mco}
	Suppose $\calP$ is a mirror poset with $2n$ vertices and its Hasse diagram has $m$ edges.  Let $k = \mco(\calP)$. Then a minimum crossing orientation of $\calP$ can be found in time $O(15^k k^4 m^3)$. In particular, MCO is fixed-parameter tractable with respect to $k = \mco(\calP)$.
\end{theorem}

\begin{proof}
	Let $\calP = (P, \leq)$ be a mirror poset with $P = \set{\rho_1, \ol{\rho}_1, \rho_2, \ol{\rho}_2,\ldots,\rho_n,\ol{\rho}_n}$.  Let $E_{\calP}$ denote the edge set of the Hasse diagram of $\mc{P}$.  We will find an orientation  of $\mc{P}$ with at most $k$ crossing edges by constructing  a $2$-SAT instance $(X, \Phi)$ from $\calP$ as follows. Let $X = \{x_1, x_2, \hdots, x_n\}$, where $x_i$ corresponds to the dual pair $\{\rho_i, \ol{\rho}_i\}$. In our reduction, a truth assignment to an $x_i$ will correspond to a choice of $\rho_i \in P^{-}$ or $\ol{\rho}_i \in P^-$. For each edge $e = \in E_\mc{P}$, we associate two clauses, $\varphi_e^1$ and $\varphi_e^2$ with $e$. Satisfying the first clause $\varphi_e^1$ corresponds to satisfying the constraint that no directed in $E_{\calP}$ can go from $P^+$ to $P^-$, while satisfying the second clause corresponds to $e$ not being a crossing edge in an orientation $(P^-, P^+)$. Formally, for $e = (\sigma_i, \sigma_j)$ with $\sigma_i \in \set{\rho_i, \ol{\rho}_i}$ and $\sigma_j \in \set{\rho_j, \ol{\rho}_j}$, define the \dft{type-1 clause} of $e$ as $\varphi_e^1 = y_i \vee z_j$ such that
	\begin{equation}\label{eqn:boolean-y}
		y_i =
		\begin{cases}
			1-x_i & \text{if } \sigma_i = \rho_i       \\
			x_i   & \text{if } \sigma_i = \ol{\rho}_i,
		\end{cases}
		\qquad\text{and}\qquad
		z_j =
		\begin{cases}
			x_j   & \text{if } \sigma_j = \rho_j       \\
			1-x_j & \text{if } \sigma_j = \ol{\rho}_j.
		\end{cases}
	\end{equation}
	Define the {\it type-2 clause} of $e$ as $\varphi_e^2 = z_i \vee y_j$. In the  $2$-SAT instance $(X, \Phi)$,  let $\Phi$ consist of $k+1$ copies of type-1 clause $\varphi_e^1$ and one copy of  type-2 clause $\varphi_e^2$ for each edge $e \in E_\calP$.

	Given a Boolean assignment $f$ of the variables in $X$, create a perfect partition $(P^-, P^+)$ of $P$ using the following rule: for each $x_i$, if $f(x_i) = 0$, let $\rho_i \in P^-$ and $ \ol{\rho}_i \in P^+$; otherwise, let $\ol{\rho_i} \in P^-$ and $ \rho_i \in P^+$.  We note that the rule is in fact a bijection from the set of all Boolean assignments of the variables in $X$ to the set of all perfect partitions of $P$.  Thus, we shall
	say that $(P^-, P^+)$ is {\it the perfect partition that  corresponds to $f$} and $f$ is the {\it Boolean assignment that corresponds to $(P^-, P^+)$}.   We now prove important properties about $f$ and $(P^-, P^+)$.


	\begin{description}
		\item[Claim 1.]  Let $f$ be a Boolean assignment of $X$ and $(P^-, P^+)$ be its corresponding perfect partition.  Let  $e = (\sigma_i, \sigma_j) \in E_\calP$. Then $f$ satisfies the type-1 clause $\varphi_e^1$ if and only if  it is {\it not} the case that $\sigma_i \in P^+$ and $\sigma_j \in P^-$.


		\item[Proof of Claim 1.]  The Boolean assignment $f$ {\it does not} satisfy
		      the clause $\varphi_e^1 = y_i \vee z_j$  if and only if  $y_i = z_j = 0$. By~(\ref{eqn:boolean-y}), $y_i = 0$ precisely when $\sigma_i = {\rho}_i$ and $f(x_i) = 1$ or $\sigma_i = \ol{\rho}_i$ and $f(x_i) = 0$.  In both cases, $\sigma_i \in P^+$ in the corresponding partition $(P^-, P^+)$ of $f$.   Similarly, by~(\ref{eqn:boolean-y}), $z_j = 0$ precisely $\sigma_j = \rho_j$ and $f(x_j) = 0$ or $\sigma_j = \ol{\rho}_j$ and $f(x_j) = 1$.  In both cases, $\sigma_j \in P^-$.  Thus, $f$ {\it does not} satisfy $\varphi_e^1 = y_i \vee z_j$ if and only if $\sigma_i \in P^+$ and $\sigma_j \in P^-$.


		\item[Claim 2.] Let $f$ be a Boolean assignment of $X$ and $(P^-, P^+)$ be its corresponding perfect partition.  Let  $e = (\sigma_i, \sigma_j) \in E_\calP$. Then $f$ satisfies the type-2 clause $\varphi_e^2$ if and only if  it is {\it not} the case that $\sigma_i \in P^-$ and $\sigma_j \in P^+$.


		\item[Proof of Claim 2.] The Boolean assignment $f$ does not satisfy $\varphi_e^2 = z_i \vee y_j$  if and only if  $z_i = y_j = 0$.   Applying the same reasoning we used to prove Claim 1, we note that $z_i = 0$ precisely when $\sigma_i \in P^-$ and $y_j = 0$ precisely when $\sigma_j \in P^+$.  Thus, $f$ does not satisfy $\varphi_e^2$ if and only if $\sigma_i \in P^-$ and $\sigma_j \in P^+$.

		\item[Claim 3.]  If a Boolean assignment $f$ of $X$ fails to satisfy at most $k$ clauses in $\Phi$ then its corresponding perfect partition $(P^-, P^+)$ is an
		      orientation of $\mc{P}$ that has at most $k$ crossing edges.

		\item[Proof of Claim 3.] Assume $f$  does not satisfy at most $k$ clauses in $\Phi$. Let  $(P^-, P^+)$ be its corresponding perfect partition.    Consider $e \in E_\mc{P}$.
		      If $f$ does not satisfy $\varphi_e^1$, the type-1 clause of $e$, then $f$ does not satisfy  $k+1$ clauses in $\Phi$ because there are $k+1$ copies of $\varphi_e^1$ in $\Phi$.  This is a contradiction.  Thus, $f$ satisfies {\it all} type-1 clauses in $\Phi$.  By Claim 1, $P^-$ is a closed subset of $\mc{P}$ so $(P^-, P^+)$ is an orientation of $\mc{P}$.  All  clauses of $\Phi$ not satisfied by $f$ must be type-2, and there are at most $k$ of them.  By Claim 2, the orientation $(P^-, P^+)$ has at most $k$ crossing edges.

		\item[Claim 4.] If $(P^-, P^+)$ is an orientation of $\mc{P}$   that has at most $k$ crossing edges, then its corresponding Boolean assignment $f$ of $X$ does not satisfy at most $k$ clauses in $\Phi$.

		\item[Proof of Claim 4.] Suppose  $(P^-, P^+)$ is an orientation of $\mc{P}$ that has at most $k$ crossing edges.  Let $f$ be the Boolean assignment of $X$ that corresponds to $(P^-, P^+)$.
		      Since $P^-$ is  a  closed subset of $\mc{P}$, none of the edges in $E_\mc{P}$ are directed from $P^+$ to $P^-$ so $f$ satisfies all type-1 clauses in $\Phi$  by Claim 1.  Additionally, because $(P^-, P^+)$ has at most $k$ crossing edges,  $f$ does not satisfy at most $k$ type-2 clauses by Claim 2.  Thus, $f$ does not satisfy at most $k$ clauses in $\Phi$.

	\end{description}

	Let $\calP$ be a mirror poset with $2n$ vertices and $m$ edges in $E_\calP$.  We now describe our FPT algorithm for  finding an orientation of $\calP$ with at most $k$ crossing edges.  First, create the 2-SAT instance $(X, \Phi)$.  The latter has $n$ variables and $m' = (k + 2) m$ clauses. Then use the FPT algorithm for Almost 2-SAT from Theorem~\ref{thm:a2sat} to find a Boolean assignment $f$ for $X$ so that the number of clauses in $\Phi$ not satisfied by $f$ is at most $k$. If $f$ exists, convert $f$ to its corresponding perfect partition $(P^-, P^+)$ and return the partition.  By Claim 3,  $(P^-, P^+)$ is an orientation of $\calP$ with at most $k$ crossing edges. Otherwise, if $f$ does not exist, return that $\calP$ has no orientation with at most $k$ crossing edges.  By Claim 4, this answer is correct.

	Constructing the 2-SAT instance from $\calP$ takes $O(n + km)$ time.  Running the FPT algorithm for Almost 2-SAT takes $O(15^k k (m')^3) = O(15^k k^4 m^3)$ time.  Finally, converting $f$ to its corresponding perfect partition $(P^-, P^+)$ takes $O(n)$ time.  Thus, our FPT algorithm runs in $O(15^k k^4 m^3)$ time, as claimed.
\end{proof}



\begin{remark}
	In Section~4.3.4 of \cite{Gusfield1989Stable}, Gusfield and Irving described two alternative ways of representing the stable matchings of a (solvable) SR instance $I$. The first one creates a graph $G'$ from the reduced rotation poset $\mc{R}'(I)$, a mirror poset.  They showed that there is a one-to-one correspondence between the maximal independent sets of $G'$ and the complete closed subsets of $\mc{R}'(I)$ and, therefore, the stable matchings of $I$.  But vertex covers are the complements of independent sets.  Thus, their result also implies that there is a bijection between the the minimal vertex covers of $G'$ and the stable matchings of $I$.  Our reduction in the proof of Theorem \ref{thm:mco-hard} goes in the opposite direction.  We create a mirror poset $\calP$ from a given graph $G$ and we show that there is a one-to-one correspondence between the vertex covers of $G$ and the orientations of $\calP$. Furthermore, the size of a vertex cover dictates the number of crossing edges in the orientation.

	In their second representation, Gusfield and Irving create a 2-SAT instance $(X', \Phi')$ from $\mc{R}'(I)$.  They showed that the satisfying assignments of $\Phi'$ are in one-to-one correspondence with the complete closed subsets of $\mc{R}'(I)$ and, therefore, the stable matchings of $I$.  In the proof of Theorem \ref{thm:mco},  our reduction also creates a 2-SAT instance $(X, \Phi)$ from a mirror poset $\calP$.  The difference between $\Phi'$ and $\Phi$ is that $\Phi'$ contains a single copy of all the type-1 clauses in $\Phi$ (along with additional (redundant) clauses corresponding to edges in the transitive closure of the Hasse diagram of $\calP$),   but $\Phi'$ does not contain any of the type-2 clauses in $\Phi$.  We use the type-2 clauses to minimize the number of crossing edges in the corresponding orientation.

\end{remark}

We now have all the ingredients  to prove Theorem~\ref{thm:main}, which we restate here.

\noindent {\bf Theorem~\ref{thm:main}} Let $I$ be an SR instance with $2n$ agents and minimum crossing distance $k$. An optimal stable matching for $I$ can be found in time $2^{O(k)} n^{O(1)}$. Thus, the optimal stable matching problem is fixed-parameter tractable with respect to minimum crossing distance.

\begin{proof}[Proof of Theorem~\ref{thm:main}]
	Let $I$ be an SR instance with $2n$ agents.  First, compute a representation of its reduced rotation poset $\mc{R}'(I)$ and do a transitive reduction to obtain the Hasse diagram of $\mc{R}'(I)$.  Next, reduce $\mc{R}'(I)$ to a 2-SAT instance and use the FPT algorithm for Almost 2-SAT to find an orientation $(R^- \cup R^+, \leq)$ of $\mc{R}'(I)$ with $k =\mco(\calP)$ crossing edges. Finally, using $I$, cost function $C$ and $R^-$, compute an optimal stable matching of $I$.   The first step takes $n^{O(1)}$ time.  According to Theorem~\ref{thm:mco} and Corollary~\ref{cor:optimal-matching}, the second and third steps take $2^{O(k)} n^{O(1)}$ and $O(n^4 \log n)$ time respectively.
 \end{proof}


%% file: conclusions.tex
\section{Conclusion}





In this paper, we have introduced the minimum crossing distance as a parameter that measures the structural similarity between an SR instance and an SM instance. Our main result establishes that finding an optimal stable matching is FPT when parameterized by the minimum crossing distance. Theorem \ref{thm:main} serves as a bridge between the two known results about the optimal stable matching problem — that it is polynomially-solvable for SM instances but NP-hard for SR instances — because the closer an SR instance is to an SM instance, the faster is the algorithm. 
We conclude by describing some interesting directions for further investigation.

\noindent {\it Typical minimum crossing distances.} What is the typical minimum crossing distance for solvable SR instances that arise in practice? For SR instances whose preference lists are generated in a restricted manner, are their  minimum crossing distances typically small or bounded? 
A deeper understanding of the relationship between preference structures and  minimum crossing distances would shed light on the practicality of our FPT algorithm for optimal stable matchings.

\noindent {\it Tighter bounds for crossing edges and maximal elements.} Corollary~\ref{cor:crossing-maximal} states that when an oriented mirror poset has $2k$ crossing edges, the corresponding median semilattice has at most $3^k$ maximal elements.  Can this upper bound be improved? Under what circumstances is the number of maximal elements exponential in $k$?

\noindent {\it More bridge results between SM and SR.} In the paper, we mentioned that computing a center stable matching \cite{ChengMS11, ChengMS16} and a robust stable matching~\cite{ChenSS19} in SM can be done in polynomial time.  If they turn out to be NP-hard problems in SR,  can  their polynomial-time algorithms be transformed into FPT algorithms for SR?  Similarly, can the FPT algorithms for balanced and sex-equal stable matchings in SM \cite{Gupta2022Treewidth}  be extended to FPT algorithms for SR?



%% file: Appendix.tex
\section{Appendix}
\label{sec:Appendix2}

	The goal of this section is to prove the two results from the main section of the paper:  Corollary \ref{corFindMatching} and 
    Lemma \ref{lemmaDiff}.   To do so, we need to discuss Irving's algorithm and the notions of tables and rotations in more detail.  We will mention some definitions and lemmas again from the main section so that the flow of the discussion is smoother.  We warn readers that we will use $T$ to refer to pseudotables and tables so that the lemmas we cite are consistent with those in \cite{Gusfield1989Stable}.

\noindent {\bf Irving's algorithm.}		Let $I$ be an SR instance. A {\it pseudotable} $T$ of $I$ contains a list for each agent.  
Initially, $T$ consists of the agents' preference lists.  Subsequently, pairs of agents 
$\{x, y\}$ are  {\it deleted}; i.e., $x$ is removed from $y$'s list and $y$ is removed from $x$'s list.  
	
	  Let  $f_T(x)$, $s_T(x)$ and $\ell_T(x)$ denote the first, second and last entries  in $x$'s list in $T$.  The pseudotable $T$ is a {\it table} of $I$  if, additionally, (i) none of the lists in $T$ are empty,  (ii) for each $x$,  $f_T(x) = y$ if and only if $\ell_T(y) = x$ and (iii) for each pair of agents $\{x, z\}$,  the pair $\{x, z\}$ is absent from $T$ if and only if $x$ prefers $\ell_T(x)$ to $z$ or $z$ prefers $\ell_T(z)$ to $x$.   
      

         Assume $T$ is a table. A cyclic sequence of ordered pairs of agents $ \rho = (x_0, y_0), (x_1, y_1), \hdots,$ $ (x_{r-1}, y_{r-1}) $	such that $y_i = f_T(x_i)$ and $y_{i+1} = s_T(x_i)$ for $i = 0, \hdots, r-1$ is a {\it rotation exposed in $T$}.   The {\it$X$-set of $\rho$} is $\{x_0, x_1, \hdots, x_{r-1}\}$ while its {\it $Y$-set} is $\{y_0, y_1, \hdots, y_{r-1} \}$.   
       To {\it eliminate} $\rho$ from $T$, all pairs $\{y_i, z\}$ such that $y_i$ prefers $x_{i-1}$ to $z$, $i = 0, \hdots, k-1$ are deleted.  The result is  $T/\rho$, which is another table as long as every agent's list is non-empty. 
       
       It is known that when some list in $T$ has two or more entries,  at least one rotation is exposed in $T$.  On the other hand, when every list in $T$ has exactly one entry, pairing each agent with the only entry in their list results in a matching because of property (ii) of a table and is stable because of property (iii).  
                
         Let us now discuss how Irving's algorithm computes a stable matching for a {\it solvable} SR instance $I$.   We note that since $I$ is solvable, every agent's list  in Irving's algorithm will never be empty.  The algorithm has two phases. We will skip over the details of the first phase because every execution of the phase yields the same output -- the {\it phase-1 table} $T_0$.   If each list in $T_0$ has exactly one entry, the algorithm outputs  the corresponding stable matching and ends.  Otherwise, it proceeds to the second phase.   While the current table $T$ has a list with two or more entries, it finds an exposed rotation and eliminates it.   The algorithm does this over and over again until each list in $T$ has only one entry.  Again, the algorithm outputs  the corresponding stable matching and ends.  When $I$ has $2n$ agents, the running time of Irving's algorithm is $O(n^2)$.  
   
   Assume table $T$ is obtained when the sequence of rotations $\langle \rho_1, \rho_2, \hdots, \rho_k \rangle$ is eliminated from $T_0$.  Suppose the same result occurs when another sequence of rotations $\langle \sigma_1, \sigma_2, \hdots, \sigma_k \rangle$ is eliminated from $T_0$.  It turns out that $\{ \rho_1, \rho_2, \hdots, \rho_k \} = \{ \sigma_1, \sigma_2, \hdots, \sigma_k \}$. Let us denote the set as $Z$.  Then we write $T = T_0/Z$. 
   The same notation applies when $T$ corresponds to a stable matching $\mu$; i.e., $\mu = T_0/Z$.

\noindent {\bf Rotations.}  We cite and prove new results on rotations.  Let $\rho$ be exposed in table $T$.  The next lemma describes the effects of eliminating $\rho$ from $T$. 
           
\begin{lemma}
(Lemma 4.2.7 in \cite{Gusfield1989Stable}.) Let $ \rho = (x_0, y_0), (x_1, y_1), \hdots, (x_{r-1}, y_{r-1}) $ be a rotation exposed in table $T$.  Then if $T/\rho$ contains no empty lists,

\noindent (i) $f_{T/\rho}(x_i) = y_{i+1}$ for $i = 0, \hdots, r-1$,

\noindent (ii) $\ell_{T/\rho}(y_i) = x_{i-1}$ for $i = 0, \hdots, r-1$,

\noindent (iii) $f_{T/\rho}(x) = f_T(x)$ for each $x$ not in the $X$-set of $\rho$ and $\ell_{T/\rho}(y) = \ell_T(y)$ for each $y$ not in the $Y$-set of $\rho$.  
\end{lemma} 

\begin{remark} 
As noted in \cite{Gusfield1989Stable},  when $\rho$ is eliminated from $T$,  the front of each $x_i$'s list (or simply $x_i$) is {\it moved one step down} from  $y_i$ to $y_{i+1}$ while the back of each $y_i$'s list (or simply $y_i$) is {\it moved one or more steps up} from $x_i$ to $x_{i-1}$.   
While each $x_i$'s list is modified from the front and each $y_i$'s list is modified from the back,  it is important to remember that there may be some agents $z$ whose lists are modified in the {\it middle} when $\rho$ is eliminated from $T$.  These are the agents that lie between $x_{i-1}$ and $x_i$ in $y_i$'s list for $i = 0, 1, \hdots, k-1$.   For such an agent $z$, $y_i$ is removed from their list.   We know $y_i$ is in the middle of $z$'s list because if $f_T(z) = y_i$ then $\ell_T(y_i) = z$ but $z$ is ahead of $x_{i}$ in $y_i$'s list.   Similarly, if $\ell_T(z) = y_i$ then $f_T(y_i) = z$ but $x_{i-1}$ is ahead of $z$ in $y_i$'s list. 
\label{remark1}
\end{remark}

Let $R(I)$ contain all the exposed rotations that can be eliminated during an execution of phase 2 of Irving’s algorithm. There are two types of rotations in $R(I)$.  A rotation  $ \rho = (x_0, y_0), (x_1, y_1), \hdots,$ $(x_{r-1}, y_{r-1}) $ is {\it non-singular} if $\overline{\rho} = (y_1, x_0), (y_2, x_1), \hdots, (y_0, x_{r-1})$ is also a rotation.  Otherwise, $\rho$ is {\it singular}.   If $\rho$ and $\overline{\rho}$ are rotations then $\overline{\rho}$ is called the {\it dual} of $\rho$ and vice versa.   Notice that the $X$-set of $\rho$ is the $Y$-set of $\overline{\rho}$ and the $Y$-set of $\rho$ is the $X$-set of $\overline{\rho}$.

\begin{theorem} (Theorem 4.3.1 in \cite{Gusfield1989Stable}.)  
Let $\mu$ be a stable matching of $I$.  If $\mu = T_0/Z$,  then $Z$ contains every singular rotation and exactly one of each dual pair of non-singular rotations of $I$. 
\label{app:thmZcontent}
\end{theorem}


Just like in the SM setting,  a poset can be created from $R(I)$.  For $\sigma, \rho \in R(I)$,  let $\sigma \leq \rho$ if for every sequence of eliminations in the phase 2 of Irving's algorithm in which $\rho$ appears, $\sigma$ appears before $\rho$.  The poset $\mc{R}(I) = (R(I), \leq)$ is the {\it rotation poset of $I$}.  Here are some properties about the ordering $\leq$:

\begin{lemma} (Lemma 4.3.7 in \cite{Gusfield1989Stable}.)
Let $\rho$ and $\sigma$ be non-singular rotations and $\tau$ be a singular rotation of $I$.  Then 

\noindent (i) $\rho$ and $\ol{\rho}$ are incomparable elements in $\mc{R}(I)$; 

\noindent (ii) $\sigma \leq \rho$ if and only if $\ol{\rho} \leq \ol{\sigma}$;

\noindent (iii) any predecessor of $\tau$ in $\mc{R}(I)$ is a singular rotation.
\label{app:lemmaRotations}
\end{lemma}

\begin{theorem} (Theorems 5.1 and 5.2 in \cite{Gusfield1988Structure}.)
There is a one-to-one correspondence between the stable matchings of $I$ and the closed subsets of $\mc{R}(I)$ that contain every singular rotation and exactly one of each dual pair of non-singular rotations of $I$. In particular, $\mu$ corresponds to the closed subset $Z$ if and only if $\mu = T_0/Z$. 
\label{mainthm3-Gusfield}
\end{theorem}

 Let $\rho = (x_0, y_0), (x_1, y_1), \hdots,$ $(x_{r-1}, y_{r-1})$ be a rotation. In \cite{Gusfield1988Structure}, Gusfield notes that for $\rho$ to be exposed in some table, all agents $z \neq y_i$ in $x_i$'s list  that appear {\it before} $y_{i+1}$ must be removed.  (He refers to them as {\it necessary eliminations}.)
 The removal of such a $z$ from $x$'s list can only happen when some rotation is eliminated. Gusfield describes this rotation in the next lemma.
 

 \begin{lemma} (Lemma 5.2 in \cite{Gusfield1988Structure}.)
  Let $\rho = (x_0, y_0), (x_1, y_1), \hdots,$ $(x_{r-1}, y_{r-1})$ be a rotation. Assume $z$ is an agent on $x_i$'s list such that $z \neq y_i$ and appears before $y_{i+1}$. 
  There exists a unique rotation $\sigma_{x_i,z}$ so that, for every sequence $\Pi$ of rotations eliminated from $T_0$ that contains $\rho$, the rotation $\sigma_{x_i,z}$ appears before $\rho$.  Moreover, $\sigma_{x_i,z}$ is the only rotation in $\Pi$ whose elimination removes $z$ from $x_i$'s list. 
  \label{lemmaPred}
 \end{lemma}

\begin{corollary} 
     Let $\rho = (x_0, y_0), (x_1, y_1), \hdots,$ $(x_{r-1}, y_{r-1})$. Assume $z$ is an agent on $x_i$'s list such that $z \neq y_i$ and appears before $y_{i+1}$.  Then the unique rotation $\sigma_{x_i,z}$ that can remove $z$ from $x_i$'s list is a predecessor of $\rho$ in  $\mc{R}(I)$.
     
\label{corPred}
\end{corollary}

Given a table $T$, another table $T'$ is a {\it subtable} of $T$ if for each agent $x$, its list in $T'$ is a sublist of the corresponding one in $T$.    A rotation $\rho$ is {\it embedded} in $T$ if there is a subtable $T'$ of $T$ so that $\rho$ is exposed in $T'$. 


\begin{lemma} (Corollary 4.2 in \cite{Gusfield1988Structure}.)
Suppose $\sigma$ is exposed in table $T$ and $\rho \neq \sigma$ is embedded in $T$.  Then $\rho$ is still embedded in $T/\sigma$ if and only if $\rho \neq \overline{\sigma}$.  
\label{lemmaEmbed}
\end{lemma}

\begin{theorem}
  Let $T = T_0/Z$.  Suppose $Z$ contains all the predecessors of $\rho$ in $\mc{R}(I)$ but not $\rho$.  Furthermore, if $\rho$ is a non-singular rotation, $\ol{\rho}$ is not in $Z$ too.  
  Then $\rho$ is exposed in $T$. 
\label{thmExposedRotation}
\end{theorem}

\begin{proof} Let $\rho = (x_0, y_0), (x_1, y_1), \hdots,$ $(x_{r-1}, y_{r-1})$.
  Suppose $\langle \sigma_1, \sigma_2, \hdots, \sigma_k \rangle$ is a sequence of rotations eliminated from $T_0$ to obtain $T$.  Hence, $Z = \{\sigma_1, \sigma_2, \hdots, \sigma_k\}$.  Let $T_i = T_{i-1}/\sigma_i$ for $i = 1, 2, \hdots, k$.  Since $\rho$ is a rotation of $I$, $\rho$ is embedded in $T_0$.  Furthermore, since none of the elements in $Z$ are equal to $\rho$ nor $\ol{\rho}$ (if $\rho$ is non-singular), it follows from Lemma \ref{lemmaEmbed} that 
  $\rho$ is embedded in $T_1, T_2, \hdots, T_{k}$.  Since $T_{k} = T$, this means that for each $x_i$, the agents $y_i$ and $y_{i+1}$ are still on $x_i$'s list in $T$.

   Now, all the predecessors of $\rho$ are in $Z$, including the ones mentioned in Corollary \ref{corPred}. Since they were eliminated from $T_0$ to obtain $T$,  it follows that for each $x_i$, $f_T(x_i) = y_i$ and $s_T(x_i) = y_{i+1}$.  Thus,  $\rho$ is exposed in $T$.  
\end{proof}

\begin{corollary}
    Let $\mu$ be a stable matching of $I$ such that $\mu = T_0/Z$. Let $\mc{R}_Z$ be the subposet of  $\mc{R}(I)$ induced by the rotations in $Z$.  Then any linear extension (or topological ordering) $\Pi$ of $\mc{R}_Z$ is a  sequence of rotations that can be eliminated from $T_0$ to produce $\mu$.
\label{corLinearExtension}
\end{corollary}

\begin{proof}
     By assumption, $\mu = T_0/Z$ so $Z$ is a closed subset of $\mc{R}(I)$ that contains every singular rotation and exactly one element from each dual pair of rotations by Theorem \ref{mainthm3-Gusfield}. 
    Let $\Pi = \langle \sigma_1, \sigma_2, \hdots, \sigma_k \rangle$ be a linear extension of $\mc{R}_Z$.   That $Z$ is a closed subset implies every predecessor of $\sigma_i$ in $\mc{R}(I)$ is in the set $Z_{i-1} = \{\sigma_1, \sigma_2, \hdots, \sigma_{i-1}\}$.  That $Z$ contains exactly one element from each dual pair of rotations of $\mc{R}(I)$ implies that when $\sigma_i$ is non-singular, $\overline{\sigma_i}$ is not in $Z_i$.  By Theorem \ref{thmExposedRotation}, if $T_{i-1}$ is a table so $T_{i-1} = T_0/Z_{i-1}$, the rotation $\sigma_i$ is exposed in $T_{i-1}$. It immediately follows that $T_i = T_{i-1}/\sigma_i$ is also a table.

     But $T_0$ is a table. By the argument above, this implies $\sigma_1$ is exposed in $T_0$ so $T_1$ is a table. Applying this reasoning over and over again, it follows that $\Pi$ is a valid sequence of rotations that can be eliminated from $T_0$.  Since $\Pi$ contains exactly the elements of $Z$, the final table corresponds to $\mu$. 
    \end{proof}

Denote a pair of agents $\{x, y\}$ a {\it stable pair} of $I$ if   $\{x, y\} \in \mu$ for some stable matching $\mu$ of $I$.  It is a {\it fixed pair} if it is in {\it every} stable matching of $I$.  The next lemma describes the relationship between the pairs that can be found in the rotations  and the stable pairs of $I$.

\begin{lemma}
(Lemma 4.4.1 in \cite{Gusfield1989Stable}.) In SR instance $I$, $\{x, y\}$ is a stable pair but not a fixed pair if and only if the pair $(x, y)$ or the pair $(y,x)$ is in a non-singular rotation of $I$.
\label{lemmaStablePartners}
\end{lemma}

In Section 6.2.2 of \cite{Gusfield1988Structure},  Gusfield sketches a procedure for computing the stable matching that corresponds to a closed subset of $\mc{R}(I)$ that contains all the singular rotations and one element of each dual pair of rotations in $O(n^2)$ time.  Corollary \ref{corFindMatching} states that a similar result holds even when only the non-singular rotations are given. 


 \noindent {\bf Corollary \ref{corFindMatching}.}
{\it Let $I$ be a solvable SR instance with $2n$ agents and let  $\mc{R}'(I)$ its reduced rotation poset.  Given a complete closed subset $S$ of $\mc{R}'(I)$,  the stable matching of $I$ that corresponds to $S$ can be obtained in $O(n^2)$  time. }

\begin{proof}
Assume $\mu$ is the stable matching that corresponds to $S$, and $\Pi$ is a sequence of rotations whose elimination from $T_0$ results in $\mu$.  For each agent $a$,  let $\mc{X}_a$ and $\mc{Y}_a$  contain the  rotations in $S$ so that $a$ is in their $X$-sets and $Y$-sets  respectively.

Suppose $\mc{X}_a \neq \emptyset$.   Let $\rho_1, \rho_2, \hdots, \rho_r \in \mc{X}_a$.  For $i = 1, \hdots, r$, let $(a, b_i) \in \rho_i$.  Without loss of generality, assume the $b_i$'s are ordered from $a$'s most preferred  to least preferred. By Corollary \ref{corPred}, $\rho_1, \rho_2, \hdots, \rho_{r-1}$ are predecessors of $\rho_r$ in $\mc{R}(I)$.  Thus, in $\Pi$, $\rho_r$ is the last rotation in $\mc{X}_a$ that is eliminated.  Let $T$ be the table just before $\rho_r$ is eliminated, and let $b_{r+1} = s_T(a)$.   Then in $T/\rho_r$, the first person on $a$'s list is $b_{r+1}$.  Moreover, as the rotations after $\rho_r$ in $\Pi$ are eliminated, $b_{r+1}$ is never removed from $a$'s list because none of the remaining rotations have $a$ is their $X$-sets.  Thus, $(a, b_{r+1}) \in \mu$.  We emphasize that $b_{r+1}$ can be obtained from $\rho_r$ itself; if $(a, b_r) = (x_i, y_i)$ in $\rho_r$, then $b_{r+1} = y_{i+1}$.   

On the other hand, suppose $\mc{Y}_a \neq \emptyset$ and $\sigma \in\mc{Y}_a$. Then $a$ in the $X$-set of $\ol{\sigma}$.  That is,  $\mc{X}_a \neq \emptyset$ so we can determine $a$'s partner in $\mu$ as outlined above.   Thus, we just have to consider agents $a$ such that $\mc{X}_a = \mc{Y}_a = \emptyset$. According to Lemma \ref{lemmaStablePartners}, for such an agent $a$, $(a,b) \in \mu$ if and only if $(a,b)$ is in {\it every} stable matching of $I$.  Run Irving's algorithm to produce some stable matching $\eta$.  Then $a$'s partners in $\mu$ and $\eta$ are the same.

There are $O(n^2)$ total pairs in the rotations in $S$.   By examining them, we can 
simultaneously determine the $\mu$-partners of all the agents that appear in some rotation in $S$ 
using the procedure above. This can be done in $O(n^2)$ time.
 For agents that are part of a fixed pair, run Irving's algorithm to determine their partners. This takes $O(n^2)$ time.  Thus, $\mu$ can be computed in $O(n^2)$ time. 
\end{proof}

\begin{lemma}
(Lemma 4.3.4 in \cite{Gusfield1989Stable}.)  If both rotations
\[
\rho = (x_0, y_0), (x_1, y_1), \hdots, (x_{r-1}, y_{r-1})
\]
and
\[
\overline{\rho} = (y_1, x_0), (y_2, x_1), \hdots, (y_0, x_{r-1})
\]
are exposed in the table $T$, then for $i = 0, \hdots, r-1$, the list of $x_i$ in $T$ contains only $y_i$ and $y_{i+1}$, and the list of $y_i$ in $T$ contains only $x_{i-1}$ and $x_i$.  The elimination of $\rho$ or $\overline{\rho}$ from $T$ reduces the list of each $x_i$ and each $y_i$ to a single entry, but affects no other list.
\label{lemmatwins}
\end{lemma}

We are now ready to prove Lemma \ref{lemmaDiff}.

 \noindent {\bf Lemma \ref{lemmaDiff}.} {\it Let $\mu$ and $\mu'$ be two stable matchings of SR instance $I$ that are adjacent in $G(I)$. 
   Let $\rho = (x_0, y_0), (x_1, y_1), \hdots,$ $(x_{r-1}, y_{r-1})$ be the (non-singular) rotation so that $S_{\mu'} - S_\mu =  \{\rho\}$.  Then
\begin{center}     
     $\mu' = \mu - \{ \{x_0, y_0\}, \{x_1, y_1\}, \hdots, \{x_{r-1}, y_{r-1}\} \} \cup \{ \{x_0, y_1\}, \{x_1, y_2\}, \hdots, \{x_{r-1}, y_{0}\} \}.$
 \end{center}
 }

 \begin{proof}
Let $Z_0$ consist of the singular rotations in $\mc{R}(I)$.  
Let $Z_\mu = Z_0 \cup S_\mu$ and $Z_{\mu'} = Z_0 \cup S_{\mu'}$ 
so $\mu = T_0/Z_\mu$ and $\mu' = T_0/Z_{\mu'}$.
Both $S_\mu$ and $S_{\mu'}$ contain exactly one element of each dual pair of rotations of $I$.  Moreover, they differ by one element with $\rho \in S_{\mu'}$.  It follows that $\ol{\rho} \in S_{\mu}$ and $S_{\mu'} - \rho = S_\mu - \ol{\rho}$.  Consequently, $Z_{\mu'} - \rho = Z_\mu - \ol{\rho}$.


In the poset $\mc{R}_{Z_{\mu'}}$, the rotation $\rho$ must be a maximal element.  Otherwise, it precedes another rotation $\sigma$.  But $\sigma \in Z_{\mu}$ while $\rho \not \in Z_{\mu}$, contradicting the fact that $Z_{\mu}$ is a closed subset of $\mc{R}(I)$.  It follows that $\mc{R}_{Z_{\mu'}}$ has a linear extension $\Pi_{\mu'}$ which has $\rho$ as its maximum element.  Similarly,  $\mc{R}_{Z_{\mu}}$ has a linear extension $\Pi_{\mu}$ with $\ol{\rho}$ as its maximum element. 
By Corollary \ref{corLinearExtension},  both $\Pi_{\mu'}$ and $\Pi_{\mu}$ are valid sequences of rotations that can be eliminated from $T_0$.  

Starting at $T_0$, eliminate the rotations in $\Pi_{\mu'}$. Denote as $T$ the table just before $\rho$ is eliminated. Notice that $T$ is the same table that is obtained when all the rotations in $\Pi_{\mu}$ except for $\ol{\rho}$ are eliminated from $T_0$. Thus, both $\rho$ and $\ol{\rho}$ are exposed in $T$. 

From Lemma \ref{lemmatwins}, the list of $x_i$ in $T$ contains $y_i, y_{i+1}$ while the list of $y_i$ contains $x_{i-1}, x_i$, for $i = 0, \hdots, r-1$.  Eliminating $\rho$ from $T$ results in $\mu'$;   it also deletes $x_i$ and $y_i$ from each other's lists.   
 Thus,  $\{\{x_i, y_{i+1}\}, i = 0, \hdots, r-1\} \subseteq \mu'$.   On the other hand, $\overline{\rho} = (y_1, x_0), (y_2, x_1), \hdots, (y_0, x_{r-1})$.  Applying the same reasoning, eliminating  $\overline{\rho}$ from $T$ results in $\mu$ and $\{\{y_i, x_i \},  i = 0, \hdots, r-1\} \subseteq \mu$.   Now, for agents $v \not \in \{x_0, x_1, \hdots, x_{r-1}, y_0, y_1, \hdots, y_{r-1}\}$, the elimination of $\rho$ nor $\overline{\rho}$ did not affect $v$'s list so it must be the case that $v$'s list consists of just one entry and it is $v$'s partner in both $\mu$ and $\mu'$.  Thus, \[\mu' = \mu -   \{\{y_i, x_i \},  i = 0, \hdots, r-1\}  \cup \{\{x_i, y_{i+1}\}, i = 0, \hdots, r-1\}.\]
\end{proof}
     


%% file: mvc-hardness.tex
\section{Hardness of \texorpdfstring{$3$}{3}-Regular MVC}

In this appendix, we provide a proof of the NP-hardness of minimum vertex cover on $3$-regular graphs (Fact~\ref{fact:mvc}). This fact appears to be well-known (cf.~\cite{Garey1979Computers}), yet we are not aware of any written proof of the fact in the literature. For completeness, we present a proof here. Our starting point is the following theorem due to Garey, Johnson, and Stockmeyer~\cite{Garey1976Some}

\begin{theorem}[\cite{Garey1976Some}]\label{thm:bounded-mvc}
	Minimum Vertex Cover is NP-hard when restricted to graphs $G = (V, E)$ where each vertex has maximum degree at most $3$.
\end{theorem}

Our proof of Fact~\ref{fact:mvc} is a two-step reduction from Theorem~\ref{thm:bounded-mvc}. We first show that MVC is hard on graphs where all vertices have degree $2$ or $3$. We then use this intermediate result to show that MVC is NP-hard on $3$-regular graphs. The basic idea for both steps is the same: define a gadget graph $H$ such that replacing each vertex of a graph $G$ with degree $1$ (or $2$) with a copy of $H$ results in a graph with no degree-1 (or degree-2) vertices such that a minimum vertex cover of $G$ can be computed from a minimum vertex cover of the resulting graph $G'$.

\begin{lemma}\label{lem:bounded-mvc}
	MVC on graphs $G = (V, E)$ with maximum degree $3$ reduces to MVC on graphs with maximum degree $3$ and minimum degree $2$.
\end{lemma}
\begin{proof}
	Let $G = (V, E)$ be a graph with maximum degree $3$, and let $v \in V$. Define the gadget graph $H_v = (\set{v_n, v_s, v, v_w}, E_v)$ where the edges form a $4$-cycle $v_n, v_w, v_s, v$ together with a vertical edge $\set{v_n, v_s}$ depicted here:
	\begin{center}
		\begin{tikzpicture}[scale=1.2,every node/.style={circle, draw, fill=white, inner sep=2pt, minimum size=18pt}]
			\node (vn) at (0,1) {$v_n$};
			\node (vs) at (0,-1) {$v_s$};
			\node (ve) at (1,0) {$v$};
			\node (vw) at (-1,0) {$v_w$};
			\draw (vn) -- (ve) -- (vs) -- (vw) -- (vn) -- (vs);
		\end{tikzpicture}
	\end{center}
	Observe that $H_v$ has maximum degree $3$ and minimum degree $2$.

	Let $G' = (V', E')$ be the graph obtained from $G$ by replacing each vertex $v$ with $\deg(v) = 1$ with a copy of $H_v$. That is if $v \in V$ with $\deg(v) = 1$, and $\set{v, w} \in E$, then $G'$ will contain the vertices $\set{v_n, v_s, v, v_w}$ together with the edge $\set{v, w}$:
	\begin{center}
		\begin{tikzpicture}[scale=1.2,every node/.style={circle, draw, fill=white, inner sep=2pt, minimum size=18pt}]
			\node (v) at (-2,0) {$v$};
			\node (w) at (-1,0) {$w$};
			\draw (v) -- (w);
			\node[draw=none] at (0,0) {$\longrightarrow$};
			\node (vn) at (2,1) {$v_n$};
			\node (vs) at (2,-1) {$v_s$};
			\node (ve) at (3,0) {$v$};
			\node (vw) at (1,0) {$v_w$};
			\node (w2) at (4,0) {$w$};
			\draw (vn) -- (ve) -- (vs) -- (vw) -- (vn) -- (vs);
			\draw (ve) -- (w2);
		\end{tikzpicture}
	\end{center}
	We first observe that the resulting graph $G'$ has minimum degree $2$ and maximum degree $3$. Suppose $G$ contained $g$ degree 1 vertices, so that $G'$ contains $g$ corresponding gadgets $H_v$.

	We claim that $G$ has a vertex cover of size at most $k$ if and only if $G'$ has a vertex cover of size at most $k + 2g$. To see this, first suppose $G$ has a vertex cover $A$ of size $k$. Let $V_1 \subseteq V$ denote the vertices in $V$ with degree $1$ in $G$, so that each $v \in V_1$ corresponds to a gadget $H_v$ in $G'$. Form the set $A'$ as
	\begin{equation}
		A' = A \cup \set{v_n, v_s \sucht v \in V_1}.
	\end{equation}
	It is clear that $A'$ is a vertex cover of size $k + 2 g$. Now suppose $A'$ is a vertex cover of $G'$ of size $k + 2 g$, and define $A = A' \cap V$. It is clear that $A$ is a vertex cover of $G$. Further, $\abs{A} \leq k$, because any vertex cover of a gadget $H_v$ must contain at least two of $\set{v_n, v_s, v_w}$.

	By the claim above, a minimum vertex cover $A'$ of $G'$ gives a minimum vertex cover $V \cap A'$ of $G$, and the lemma follows.
\end{proof}

Using Lemma~\ref{lem:bounded-mvc}, we can prove Fact~\ref{fact:mvc}.

\begin{proof}[Proof of Fact~\ref{fact:mvc}]
	By Lemma~\ref{lem:bounded-mvc}, it suffices to show that MVC on graphs with maximum degree $3$ and minimum degree $2$ reduces to MVC on $3$-regular graphs. To this end, suppose $G = (V, E)$ is a graph with maximum degree $3$ and minimum degree $2$. Let $V_2 = \set{v \in V \sucht \deg(v) = 2}$ denote the set of vertices with degree $2$. For each $v \in V$, define the gadget $H_v = (\set{v_n, v_s, v_e, v_w}, E_v)$ to be the $4$-cycle with a vertical edge as before:
	\begin{center}
		\begin{tikzpicture}[scale=1.2,every node/.style={circle, draw, fill=white, inner sep=2pt, minimum size=18pt}]
			\node (vn) at (0,1) {$v_n$};
			\node (vs) at (0,-1) {$v_s$};
			\node (ve) at (1,0) {$v_e$};
			\node (vw) at (-1,0) {$v_w$};
			\draw (vn) -- (ve) -- (vs) -- (vw) -- (vn) -- (vs);
		\end{tikzpicture}
	\end{center}
	Form the graph $G'$ by replacing each $v \in V_2$ with $H_v$ as follows: if $v$ has neighbors $u$ and $w$, then replace $v$ with $V_v$, add the edges $E_v$ along with edges $\set{u, v_w}$ and $\set{v_e, w}$:
	\begin{center}
		\begin{tikzpicture}[scale=1.2,every node/.style={circle, draw, fill=white, inner sep=2pt, minimum size=18pt}]
			\node (u) at (-3,0) {$u$};
			\node (v) at (-2,0) {$v$};
			\node (w) at (-1,0) {$w$};
			\draw (u) -- (v);
			\draw (v) -- (w);
			\node[draw=none] at (0,0) {$\longrightarrow$};
			\node (u2) at (1,0) {$u$};
			\node (vn) at (3,1) {$v_n$};
			\node (vs) at (3,-1) {$v_s$};
			\node (ve) at (4,0) {$v_e$};
			\node (vw) at (2,0) {$v_w$};
			\node (w2) at (5,0) {$w$};
			\draw (vn) -- (ve) -- (vs) -- (vw) -- (vn) -- (vs);
			\draw (ve) -- (w2);
			\draw (u2) -- (vw);
		\end{tikzpicture}
	\end{center}
	It is clear that $G'$ is a $3$-regular graph.

	As before, wee claim that $G$ admits a vertex cover of size at most $k$ if and only if $G'$ admits a vertex cover of size at most $k + 2 g$, where $g = \abs{V_2}$ is the number of gadgets added to $G$ to form $G'$. First, suppose $A \subset V$ is a vertex cover of $G$ of size $k$, and let $A_2 = A \cap V_2$. Now form $A'$ as follows:
	\begin{equation*}
		A' = (A \setminus A_2) \cup \set{v_w, v_n, v_e \sucht v \in A_2} \cup \set{v_n, v_s \sucht v \in V_2 \setminus A_2}.
	\end{equation*}
	That is, for each vertex $v \in A_2$, we add the three vertices $v_w, v_n, v_e$ to $A'$, and for each vertex $v \in V_2 \setminus A_2$, we add the vertices $v_n, v_s$ to $A'$. Note that $\abs{A'} = \abs{A} + 2 g$. Further, $A'$ is a vertex cover: if $v \in A_2$, then adding $v_w$ and $v_e$ to $A'$ ensures that the edges $\set{u, v_w}$ and $\set{v_e, w}$ are covered, while if $v \in V_2 \setminus A_2$, then these edges are already covered in $A$ by $u$ and $w$, respectively.

	Conversely, suppose $G'$ admits a vertex cover $C'$ of size $k + 2g$. Form the set $A$ as follows:
	\begin{equation}\label{eqn:3-mvc}
		A = (A' \cap V) \cup \set{v \in V_2 \sucht v_w \in A' \text{ or } v_e \in A'}.
	\end{equation}
	That is, we take $A$ to be $A'$'s restriction to $V$, union any vertices $v$ such that $v_w$ or $v_e$ is contained in $A'$. To see that $A$ is a vertex cover of $G$, note that including $\set{v \in V_2 \sucht v_w \in A' \text{ or } v_e \in A'}$ in $A$ ensures that all edges incident to vertices $v \in V_2$ are covered, while all other edges were covered because $A'$ was a vertex cover of $G'$. We note that $\abs{A} \leq \abs{A'} - 2g$. This is because any vertex cover $A'$ must contain at least $3$ elements from $V_v$ unless $A' \cap V_v = \set{v_n, v_s}$, in which case $v$ is not added to $A$. Thus, each $H_v$ corresponds to a net loss of at least $2$ vertices from $A'$.

	Finally, by the claim, if $A'$ is a minimum vertex cover for $G'$, then the $A$ as in~(\ref{eqn:3-mvc}) is a minimum vertex cover for $G$.
\end{proof}